\documentclass[
prd,
twocolumn,
superscriptaddress,
nofootinbib,
 amsmath
 amssymb,
 aps,
 eqsecnum,
]{revtex4-1}
\usepackage[utf8]{inputenc}
\usepackage{mathtools}
\usepackage{amssymb}
\usepackage{appendix}
\usepackage{amsmath}
\usepackage{amsthm}
\usepackage{graphicx}
\usepackage{physics}
\usepackage{braket}
\usepackage{pythontex}
\graphicspath{{./images/}}
\usepackage{graphicx}
\usepackage{dcolumn}
\usepackage{bm}
\usepackage{hyperref}
\hypersetup{
    colorlinks,
    linkcolor = {blue!100!black},
    citecolor = {blue!100!black},
    urlcolor = {blue!100!black}
}

\let\originalleft\left
\let\originalright\right
\renewcommand{\left}{\mathopen{}\mathclose\bgroup\originalleft}
\renewcommand{\right}{\aftergroup\egroup\originalright}

\theoremstyle{plain}
\newtheorem{theorem}{Theorem}

\theoremstyle{remark}

\DeclareMathOperator{\sinc}{sinc}

\renewcommand{\Im}{\operatorname{Im}}

\makeatletter

\newcommand{\ringplus}{\mathbin{\text{\@ringplus}}}

\newcommand{\@ringplus}{%
  \ooalign{\hidewidth\raise1.3ex\hbox{\tiny$\circ$}\hidewidth\cr$\m@th+$\cr}%
}

\newcommand{\ringminus}{\mathbin{\text{\@ringminus}}}

\newcommand{\@ringminus}{%
  \ooalign{\hidewidth\raise0.9ex\hbox{\tiny$\circ$}\hidewidth\cr$\m@th-$\cr}%
}
\makeatother

\newcommand{\tp}[0]{\mathrm{T}}

\DeclareFontFamily{U}{wncy}{}
\DeclareFontShape{U}{wncy}{m}{n}{<->wncyr10}{}
\DeclareSymbolFont{mcy}{U}{wncy}{m}{n}
\DeclareMathSymbol{\Sh}{\mathord}{mcy}{"58}

\providecommand{\abs}{}
\renewcommand{\abs}[1]{\left\lvert{#1}\right\rvert}

\newcommand{\reals}[0]{\mathbb{R}}
\newcommand{\complex}[0]{\mathbb{C}}

\newcommand{\integers}[0]{\mathbb{Z}}

\providecommand{\op}[1]{}
\renewcommand{\op}[1]{\hat{#1}}

\newcommand{\opvec}[1]{\op{\vec{#1}}}

\newcommand{\mat}[1]{\bm{\mathrm{#1}}}

\renewcommand{\vec}[1]{\bm{\mathrm{#1}}}

\newcommand{\partialBL}{\partial^\mathrm{(BL)}}

\newcommand{\blk}{\color{black}}

\usepackage{graphicx}
\graphicspath{{./figures/}}

\DeclareMathOperator{\Li}{Li}

\allowdisplaybreaks[4]

\newcommand{\valforhbar}{}

\begin{document}

\title{Quantum lattice models that preserve continuous translation symmetry}

\author{Dominic G. Lewis}
 \email{dominic.lewis@student.rmit.edu.au}
\affiliation{%
 Center for Quantum Computation and Communication Technology, School of Science, RMIT University, Melbourne, Victoria 3000, Australia
}%
\author{Achim Kempf}
 \affiliation{Institute for Quantum Computing, University of Waterloo, Waterloo, Ontario N2L 3G1, Canada
 }%

\affiliation{%
 Department of Applied Mathematics, University of Waterloo, Waterloo, Ontario N2L 3G1, Canada
}%

\author{Nicolas C. Menicucci}%
  \email{nicolas.menicucci@rmit.edu.au}
\affiliation{%
 Center for Quantum Computation and Communication Technology, School of Science, RMIT University, Melbourne, Victoria 3000, Australia
}%

\date{November, 2023}

\begin{abstract} 
     Bandlimited approaches to quantum field theory offer the tantalizing possibility of working with fields that are simultaneously both continuous and discrete via the Shannon Sampling Theorem from signal processing. Conflicting assumptions in general relativity (smooth spacetime) and quantum field theory (high-energy deviations from low-energy emergent smoothness) motivate the use of such an appealing analytical tool that could thread the needle to meet both requirements.
     Bandlimited continuous quantum fields are isomorphic to lattice theories---yet without requiring a fixed lattice. Any lattice with a required minimum spacing can be used. This is an isomorphism that avoids
     taking the limit of the lattice spacing going to zero. In this work, we explore the consequences of this isomorphism, including the emergence of effectively continuous symmetries in quantum lattice theories. One obtains conserved lattice observables for these continuous (Noether) symmetries, as well as a duality of locality from the two perspectives. We expect this work and its extensions to provide useful tools for considering numerical lattice models of continuous quantum fields (e.g.,~lattice gauge theories) arising from the availability of discreteness without a fixed lattice, as well as offering new insights into emergent continuous symmetries in lattice models and possible laboratory demonstrations of these phenomena.
\end{abstract}
\maketitle

\section{Introduction \label{introduction}}

The combination of effects of quantum mechanics with those of general relativity is generally expected to imply the existence of a minimum length scale~\cite{UpdatedHossenfelder2013, UpdatedWheeler1957, UpdatedGaray1995, PhysRevD.52.1108, UpdatedMead1964}: in essence, higher precision in measurements in position leads to higher uncertainty in momentum and thus higher uncertainty in curvature. Inevitably, a point will be reached where the uncertainty in curvature fundamentally restricts our ability to precisely measure position, implying that the notion of distance breaks down in some way at a minimum length scale. 
It is estimated that this minimum length~$\Delta x$ is on the order of the Planck length $l_p = \sqrt{\frac{G\hbar}{c^3}}\approx 10^{-35}~\mathrm{m}$.

This notion of minimum length scale also appears in quantum field theory. Despite being some of the most successful theories of the last century for matching predictions with experimental results, our current best quantum field theories~(QFTs) break down at the smallest scales and give nonsensical results, which are resolvable in some circumstances using renormalization~\cite{UpdatedPeskin2018}. This indicates that it may be natural to assume that spacetime, at its smalles scales, is not arbitrarily smooth and distances cannot be resolved to arbitrary precision. Most of the current approaches to quantum gravity acknowledge some sort of small-scale structure or indefiniteness to spacetime at the smallest scales, including string theory \cite{UlfDanielsson_2001, KONISHI1990276, UpdatedSusskind1994}, spinfoam models of loop quantum gravity \cite{UpdatedPerez2003, UpdatedAshtekar2004}, causal set theory \cite{UpdatedSorkin2005}, and others (see \cite{UpdatedHossenfelder2013} for a recent review).  \blk

\begin{figure}
    \centering
    \includegraphics[width=\linewidth]{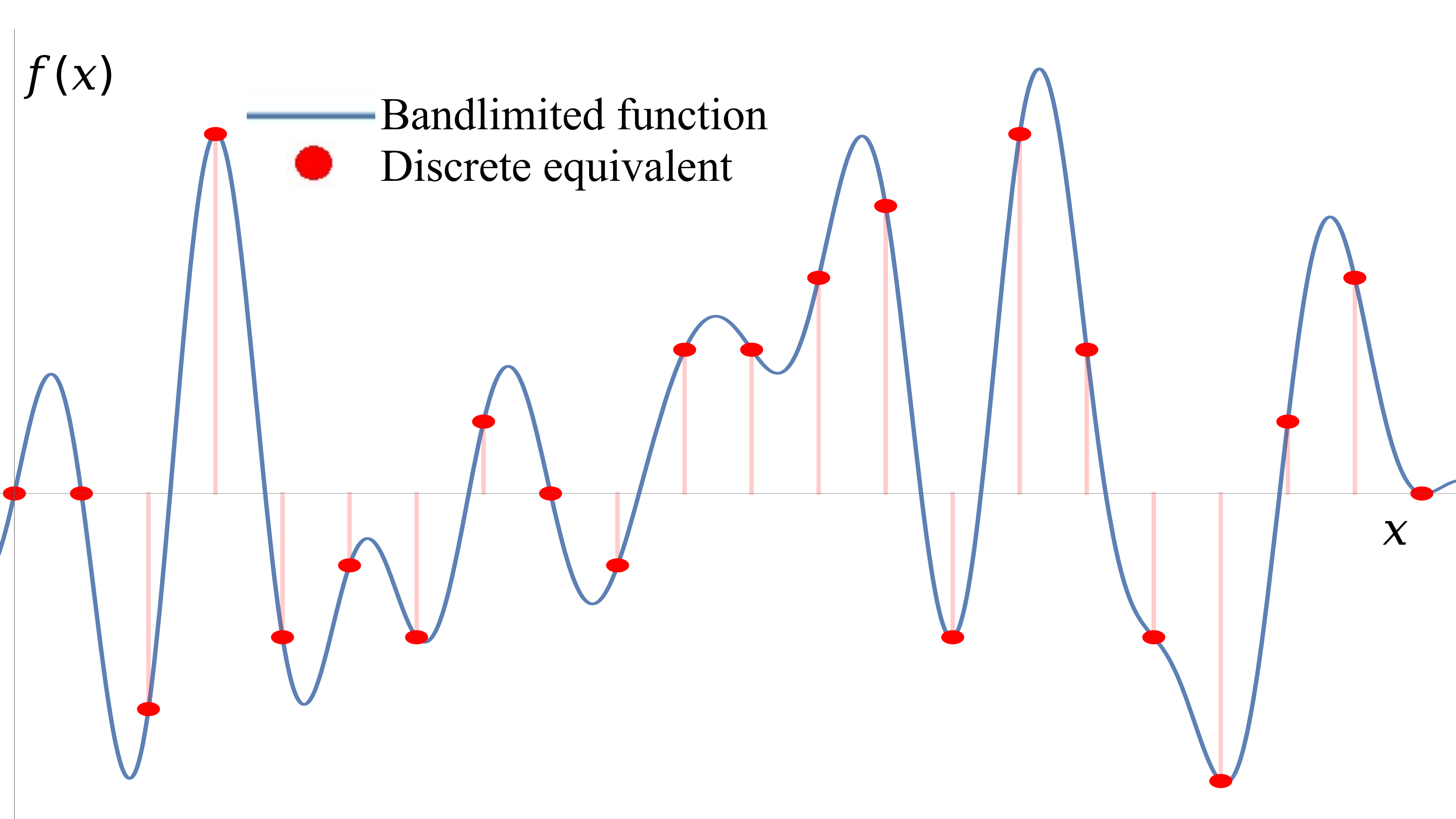}
    \caption{Example of a bandlimited function (blue) and its discrete equivalent (red) produced by sampling the function on an evenly spaced lattice.}
    \label{fig:bandlimitedfunctionsampling}
\end{figure}
Given this apparent breakdown of continuous quantum fields at the smallest scales and the implication of a minimum length scale discussed previously, it would make sense to simply switch to a lattice field theory when considering systems at very small scales. This has been done previously with considerable success in lattice QCD where switching to a lattice theory aids in computation \cite{UpdatedMuroya2003, UpdatedDrell1974}. That being said, a quantum field theory based purely on a lattice is still limited by the fact that it is difficult to merge with general relativity. Einstein's equations of relativity require spacetime to be a continuous and smooth manifold at all scales, meaning that it is unlikely for a lattice theory to be compatible.

This notion that spacetime may be simultaneously continuous and discrete presents an apparent contradiction. However, a method of representing a mathematical object such as a function or field as equivalently continuous and discrete while simultaneously introducing a minimum length scale can be found in information theory. 

 Consider a classical signal whose Fourier transform has compact support in frequency.
Such a signal is called \emph{bandlimited}, with the width of the support in frequency called the bandwidth. A bandlimited signal can be reconstructed perfectly---i.e.,~with zero loss of information in the absence of noise---from any sufficiently dense lattice of sample points~\cite{UpdatedShannon1948, lerman2015shannon}. The same holds for spatial functions bandlimited in spatial frequency.
Notably, for a spatial function~$f(x)$ with symmetric frequency support about zero, i.e.,~for $|k|< \Omega$, where $k$ is the spatial frequency (wavenumber), the continuous function can be reconstructed from its values on a uniform lattice of spacing $\Delta x = \frac{\pi}{\Omega}$, where $\Omega$ is the \emph{bandlimit} (i.e.,~maximum spatial frequency). $\Delta x$~is called the \emph{Nyquist spacing} and is the maximum lattice spacing from which a continuous bandlimited function can be perfectly reconstructed. Note that under these conditions, the reconstruction of the function from its samples is perfect---i.e., it is not an approximation~\cite{UpdatedPye2015, UpdatedShannon1948, lerman2015shannon}.

The same is true for a field, whether classical or quantum~\cite{PhysRevLett.85.2873, PhysRevLett.103.231301, PhysRevD.74.103517, PhysRevD.63.083514, UpdatedGrimmer2022}. In fact, in QFT it is already possible to describe an inability to resolve distances to arbitrary precision. This is done using  bandlimited fields---i.e.,~those with a hard UV cutoff. It was shown by Pye, Donelly, and Kempf~\cite{UpdatedPye2015} that the spatial profile of a continuous bandlimited quantum field reconstructed from a lattice will be incompressible below the Nyquist spacing, meaning that a reconstruction of a bandlimited field  from a  denser lattice than that prescribed by the Nyquist condition is completely equivalent to a reconstruction of the same field from a lattice at the Nyquist spacing.

Many common quantum field theories are represented in terms of their Fourier transform~\cite{UpdatedPeskin2018, ng2009introduction}, making the act of imposing a UV cutoff as simple as limiting the domain of the Fourier representation of the fields.  A fiducial cutoff may be applied as part of a renormalization process\blk~\cite{UpdatedPeskin2018}. Instead, we consider here that nature itself possesses a finite UV cutoff, analogous to the maximum frequency with which a human ear can perceive sound. Where frequencies higher than this are not detected and thus are `cut off'~\cite{UpdatedSumarno2019} naturally and do not need to be considered when transmitting or receiving audio signals. In this cut-off universe, frequencies beyond the cut-off cannot interact with nor be detected by anything existing below the cut-off. As a result, those frequencies can be treated as if they do not exist at all. Thus, we take the notion of a UV cut-off to quantum fields and treat it as a physical property of our universe as opposed to a tool of re-normalization.\blk 

Bandlimited quantum fields have been investigated to some extent over the last two decades~\cite{UpdatedHenderson2020, kempf1999generalized, UpdatedKempf2013, UpdatedKempf2021, UpdatedPye2015, UpdatedPye2019}, with much of said investigation being done by one of us (Kempf)~\cite{UpdatedPye2015, kempf1999generalized, PhysRevLett.85.2873, PhysRevLett.103.231301, UpdatedKempf2013, UpdatedKempf2021, PhysRevD.69.124014}.  These results have shown that a \blk covariantly bandlimited quantum field in Minkowski space can be achieved by cutting off the spectrum of the D'Alembertian~\cite{UpdatedKempf2013}. Unlike in the non-covariant case, a field featuring such a covariant UV cut-off  may still have spatial modes at or beyond the `bandlimit'. However, the \emph{temporal} bandwidth for such modes are suppressed to the extent that they are essentially frozen in time and do not significantly affect the dynamics of the spatial modes below the cut-off. For a theoretical study and applications in cosmology, see \cite{UpdatedKempf2013, ChatwinDaviesNaturalCovariant, Kempfcovariantinformationdensity, chatwindavieskempfCovariantpredictions}. 
That being said, sampling theory for a covariantly cut-off field requires functional analitic tools and attempting to reconstruct such a field from a lattice after boosts becomes unwieldy \cite{UpdatedKempf2013}.

Here we consider only a noncovariant 
 momentum  cutoff as a limiting case of the covariant cutoff \cite{UpdatedKempf2013}.  
Imposing a hard, non-covariant, UV cutoff to quantum fields has the distinct benefit of creating a bridge between the continuous and the discrete with a sampling theory that is relatively easy to  handle mathematically and whose qualitative results should still have some analog in the fully covariant case in the nonrelativistic limit \cite{UpdatedPye2015, UpdatedKempf2013, UpdatedPye2023}.  
In this paper, we use the framework of non-covariantly bandlimited quantum fields to provide a toolbox for the representation of lattice theories as continuous fields and vice versa. We investigate the emergence of continuous symmetries in the lattice framework, discuss the arbitrariness of the particular chosen lattice, and discuss the duality of locality and non-locality when transforming between the two pictures. We limit our discussion to $(1+1)$-dimensional scalar fields for simplicity of presentation, leaving extensions to other fields to future work.
In Sec.~\ref{sec:bandlimitedKG}, we outline the notion of bandlimited quantum field theories and introduce the mathematical toolbox that comes with bandlimitation, including discrete equivalents to traditionally continuous operations that are lattice independent despite being discrete.

We determine the Hamiltonian for the bandlimited free Klein-Gordon field as a
lattice theory using this toolbox. In Sec.~\ref{section:DiscreteQFT}, we determine the Hamiltonian for the continuous, bandlimited field representing the discrete harmonic chain (i.e.,~the usual example with nearest-neighbor coupling only) and discuss the differences between the harmonic chain field and the bandlimited Klein-Gordon field in both the continuum and on the lattice, despite these fields agreeing in the limit~$\Delta x \to 0^+$. Additionally, we use the lifting of discrete fields to a continuum to introduce the emergence of fully continuous translational symmetry on lattice theories and its associated Noether observable, the total field momentum. Finally, we conclude with possible uses of this toolbox of results in Sec.~\ref{sec:outlook}.


\section{Equivalent discrete representation of a continuous bandlimited quantum field}\label{sec:bandlimitedKG}


\subsection{Bandlimited quantum fields}

First, we will investigate some of the effects that a UV cut-off has on the form and behavior of a quantum field. To do so we will introduce some of the mathematical tools associated with bandlimited functions and how they can be extended to bandlimited quantum fields. We begin with a  bandlimited function. As stated previously, when a function's Fourier transform is supported only over a  bounded frequency interval~\blk $D$, given by
\begin{align}
\label{eq:BLfwithD}
    f(x) = \int_D \frac{dk}{2\pi}\tilde{f}(k)e^{-ikx},
\end{align}
that function is called bandlimited.
If the region~$D$ is an open interval symmetric about the origin, $D = (-\Omega, \Omega)$,%
\footnote{We restrict to an open interval despite the fact that $f(x) = \cos(\Omega x)$ can be perfectly reconstructed from a lattice symmetric about the origin. We do this because $f(x) = \sin(\Omega x)$ cannot be reconstructed using this lattice (the reconstruction gives $f(x)=0$ in this case), which means that the claim about the lattice being arbitrary fails to hold for signals at the bandlimit~$\Omega$. Since we want to preserve this arbitrariness of lattice, we excise the boundary points from the spectral support~$D$. This issue has been addressed in previous work by one of us in Ref. \cite{PhysRevLett.85.2873}.}
the function can be reconstructed perfectly from discrete samples on any uniform lattice under the condition that said lattice has spacing less than or equal to $\frac{\pi}{\Omega}$. The lattice of spacing $\frac{\pi}{\Omega}$ is called the Nyquist lattice and the reconstruction is given by the Shannon reconstruction formula~\cite{UpdatedShannon1948, UpdatedPye2015}:

\begin{align}
\label{shannon}
    f(x) = \sum_{j\in\integers} f(x_j) \sinc_\pi \left(\frac{x-x_j}{\Delta x}\right)
    ,
\end{align}
where%
\footnote{In some works, $\sinc_\pi$ is called the \emph{normalized} sinc function and is defined without the $\pi$ subscript. We will keep the notation introduced here to separate the definitions of $\sinc$ and $\sinc_\pi$ and thereby avoid any confusion.} 
\begin{align}
    \sinc_\pi(x) &\coloneqq \sinc(\pi x)
    ,
    \label{sincpi}
\\
    \sinc(x) &\coloneqq 
    \begin{cases}
    \frac{\sin(x)}{x} & \text{if $x \neq 0$,} \\
    1 & \text{if $x = 0$,}
    \end{cases}
\intertext{and}
x_j 
&\coloneqq
    j\Delta x + b  
, 
\label{xjdef}
\end{align}

for some $b\in [0, \Delta x)$.

 When $x$ is restricted to integers, $\sinc_\pi$ reduces to the Kronecker delta:
\begin{align}
    (\sinc_\pi) \rvert_\integers (n) &= \delta_{n,0}
    .\label{sinclattice}
\end{align}
Note that Eq.~\eqref{shannon} is an equality---not an approximation. 
A proof of this can be found in Appendix~\ref{app:shannonproof}. 

Additionally, note that from the freedom of choice of $b$, the continuous representation of $f$ is equivalent to an infinite number of discrete representations, where different values for $b$ define different lattices from which to sample $f$.

Finally, we note that any bandlimited function is an entire function, i.e., it is analytic in the entire complex plane~\cite{PhysRevD.69.124014, UpdatedBoyd2003}.\blk

Turning now to quantum fields~\cite{PhysRevLett.85.2873, PhysRevD.69.124014, UpdatedPye2015}, we note that a bandlimtied scalar field can be written as an infinite sum of discrete samples on a uniform lattice of spacing $\Delta x$ using Shannon reconstruction:
\begin{align}
    \label{BLphidef}
    \op{\phi}(x) =  \int_{-\Omega}^\Omega \frac{dk}{2\pi}\op{\tilde{\phi}}(k)e^{-ikx} = \sum_{j\in\integers} \op{\phi}(x_j)\sinc_\pi\left(\frac{x-x_j}{\Delta x}\right)
    .
\end{align}

For a bandlimited field, we write the field and its conjugate momentum following Eq.~(\ref{BLphidef}):
\begin{align}
    \op{\phi}(x) &= \sum_{j\in\integers} \op{\phi}( x_j ) \sinc_\pi \left(\frac{x-x_j}{\Delta x}\right),\label{phiBL}\\
    \op{\pi}(x) &= \sum_{j\in\integers} \op{\pi}( x_j )\sinc_\pi \left(\frac{x-x_j}{\Delta x}\right),\label{piBL}
\end{align}
where $\op{\phi}$ and $\op{\pi}$ satisfy the modified commutation relation given by~\cite{UpdatedPye2015}
\begin{align}
    \left[\op{\phi}(x),\op{\pi}(y)\right] 
    &=\frac{i\valforhbar}{\Delta x}\sinc_\pi \left(\frac{x-y}{\Delta x}\right)\label{contcom}
    ,
\end{align}
where we use $\hbar = 1$ throughout. 
To keep the notation compact, we define lattice operators $\op{q}$ and $\op{p}$ in terms of $\op{\phi}$ and $\op{\pi}$, respectively:
\begin{subequations}
\label{qpdef}
\begin{align}
    \op{q}_j &\coloneqq \op{\phi}(x_j),\label{qdef}\\
    \op{p}_j &\coloneqq \Delta x\, \op{\pi}(x_j).\label{pdef}
\end{align}
\end{subequations}
Evaluation of Eq.~(\ref{contcom}) on a lattice of spacing $\Delta x$ gives the commutation relation for $\op{q}$ and $\op{p}$ through Eq.(\ref{sinclattice}).
\begin{align}
    \left[\op{q}_j , \op{p}_k\right] = i\valforhbar \delta_{j,k}. \label{disccom}
\end{align}

In Eq.~(\ref{contcom}) it can be seen that when \(x-y\) is a multiple of \(\Delta x\)---i.e., the fields are evaluated at points from the same lattice---the \(\sinc_\pi\) in Eq.~(\ref{contcom}) simplifies to the Kronecker delta in Eq.~(\ref{disccom}), and thus Eq.~(\ref{contcom}) reduces to Eq.~(\ref{disccom}). However, if $x$ and $y$ are not spaced by a multiple of $\Delta x$, the commutator does not vanish for non-local interactions. Note that in the continuum limit of \(\Delta x \rightarrow 0^+\), $(\Delta x)^{-1} \sinc_\pi [(\Delta x)^{-1}(x-y)]$ acts as a nascent delta function, and the commutator reduces to
\begin{align}
    \lim_{\Delta x \rightarrow 0^+}\left[\op{\phi}(x),\op{\pi}(y)\right] = i\valforhbar \delta(x-y).
\end{align}
Note that taking this continuum limit is 
the equivalent of allowing the UV cutoff to extend to infinity.

One can also use Shannon's sampling theorem to transform from one lattice of samples to another, so long as the spacing remains consistent. We have
\begin{align}
    \op{q}_{k'}
     =
     \op{\phi}(x_{k'}) 
     =
     \sum_{j\in\integers}
     \op{q}_j\sinc_\pi\left(\frac{x_{k'}-x_j}{\Delta x}\right), 
     \label{latticetransformation}
\end{align}
where $k'$ indicates that the samples at $x_{k'}$ are set on an entirely different lattice to the $x_j$'s. As such, one can transform from one lattice to another using only linear combinations of the samples on the first lattice.

The derivative of a bandlimited function is also bandlimited%
\footnote{In the frequency domain, the spatial derivative is multiplication by $ik$, which leaves the spectral support unchanged.}
and thus has a discrete equivalent. As a result, for bandlimited functions, the derivative can be written as a linear map $D$ from the function values on a given Nyquist lattice to the values of the function's derivative on the same lattice. 
We have
\begin{align}
    \label{deriv1sumdef}
    f'(x_j) = \sum_{k\in\integers}D_{jk}f(x_k).
\end{align}
The elements $D_{jk}$ of this linear map are determined by taking the derivative of the Shannon reconstruction formula and evaluating on the lattice (see Appendix \ref{app:BLderivs} for further details) and are given by
\begin{align}
    \label{derivativematrixdef}
    D_{jk} 
    =
    \begin{cases}\displaystyle
        \frac{1}{\Delta x} 
        \frac{(-1)^{j-k}}{j-k}
        &
        \textrm{if } j\neq k
        \\\displaystyle
        \vphantom{\frac{1}{\Delta x} }
        0
        &
        \textrm{if }  j=k.
    \end{cases}
\end{align}
We can extend on this notion by considering a vector $\vec{f}$ with elements $f(x_k)$ that are a bandlimited function's samples on a Nyquist lattice. Acting on it with a matrix $\mat{D}$ whose elements are $D_{jk}$ produces
\begin{align}
    \vec{f'}\coloneqq \mat{D}\vec{f},
\end{align}
where $\vec f'$ is a vector of the samples of the derivative of $f$ on the same Nyquist lattice. 
Note that the bandlimited derivative has indentical form to the SLAC derivative in lattice QCD \cite{UpdatedDrell1974, BERGNER2008946} and the infinite order finite difference stencil approximation to the derivative \cite{UpdatedFornberg2022}. We will discuss this further in section \ref{sec:outlook}. 
Like the first derivative, the second derivative of a bandlimited function is also bandlimited, and the same treatment as was done for the first derivative can be done for the second. Specifically, we write the second derivative as a linear map $D_{(2)}$ from a function's samples on a Nyquist lattice to the samples of the function's second derivative on the same lattice. We have
\begin{align}
    \label{deriv2sumdef}
    f''(x_j) 
    = 
    \sum_{k\in\integers}{[D_{(2)}}]_{jk} f(x_k),
\end{align}
where ${[D_{(2)}]}_{jk}$ can be determined by taking the second derivative of the Shannon reconstruction formula or by squaring the matrix $\mat{D}$ (see Appendix \ref{app:BLsecondderivs} for details) and are given by
\begin{align}
    {[D_{(2)}]}_{jk}
    \coloneqq
    \begin{cases}
    \displaystyle
        -\frac{\pi^2}{3(\Delta x)^2}
        &
        \textrm{if }j = k
        \\
    \displaystyle
        -\frac{2}{(\Delta x)^2}
        \frac{(-1)^{j-k}}{(j-k)^2}
        &
        \textrm{if }j\neq k.
    \end{cases}
\end{align}
Like with the first derivative, we can also write this as a matrix $\mat{D}_{(2)}$ acting on a vector of function samples to give a vector of samples of the function's second derivative
\begin{align}
    \vec{f''}
    \coloneqq
    \mat{D}_{(2)}\vec{f} 
    = \mat{D}^2\vec{f}.
\end{align}
Since both $\mat{D}$ and $\mat{D}_{(2)}$ are Toeplitz matrices, i.e.,~containing elements dependent on only $(j-k)$, we can reindex the sums in Eqs.~\eqref{deriv1sumdef} and~\eqref{deriv2sumdef} such that $k-j=m$ to write
\begin{align}
    \partial_x f(x)\Big|_{x=x_j} 
    &= 
    \frac{-1}{\Delta x}
    \sum_{n\neq 0}
    \frac{(-1)^n}{n}
    f(x_j+n\Delta x),
\\
    \partial_{xx} f(x)\Big|_{x=x_j}
    &= 
    -\frac{\pi^2}{3(\Delta x)^2}
    f(x_j) 
    \nonumber
\\
    &
    \hspace{0.4cm}
    -\frac{2}{(\Delta x)^2}
    \sum_{n\neq0}
    \frac{(-1)^n}{n^2}
    f(x_j + n\Delta x).
\end{align}
 For bandlimited functions, these linear maps can act on the continuous representation of the function as well as the discrete one, acting as equivalents to the standard notion of the derivative.
These \emph{bandlimited derivatives}, indicated with $\mathrm{^{(BL)}}$, are shown in Eqs.~\eqref{eq:BLderivdef} and~\eqref{eq:BL2derivdef} as equivalents to the first and second derivative (note again that the method of calculation of these bandlimited derivatives can be found in Appendices~\ref{app:BLderivs} and~\ref{app:BLsecondderivs}):
\begin{align}
    \partialBL_x
&=
    \frac{-1}{\Delta x}\sum_{n\neq0}\frac{(-e^{\Delta x \partial_x})^n}{n},
    \label{eq:BLderivdef}
\\
    \partialBL_{xx}
&=
    -\frac{\pi^2}{3(\Delta x)^2}-\frac{2}{(\Delta x)^2}\sum_{n\neq 0}\frac{(-e^{\Delta x\partial_x})^{n}}{n^2}
    ,
    \label{eq:BL2derivdef}
\end{align}
where $\partial_{xx}$ is the second $x$ derivative, and $e^{a \partial_x}$ is a displacement operator, viz.
\begin{align}
\label{eq:transop}
    e^{a \partial_x} f(x) = f(x+a),
\end{align}
in which the left-hand side is a compact representation of the Taylor series of $f$ at $x$ (expand the exponential as a power series), and the right-hand side holds when $f$ is analytic---which is true for any bandlimited function~\cite{PhysRevD.69.124014, UpdatedBoyd2003, UpdatedPollock2012}, although these operators can be applied to any analytic function. Note also that $(e^{a \partial_x})^n = e^{na \partial_x}$, which simply says that $n$ iterations of the same displacement is the same as displacing by $n$ times the original amount.

Applying the bandlimited first and second derivative operators to a bandlimited quantum field~$\op \phi$ in the continuous representation gives
\begin{align}
    \partialBL_x \op \phi(x)
&=
    \frac{-1}{\Delta x}\sum_{n\neq0}
    \frac{(-1)^n}{n}
    \op \phi (x + n\Delta x)
    ,
    \label{eq:BLderiveval}
\\
    \partialBL_{xx} \op \phi(x)
&=
    -
    \frac{\pi^2}{3(\Delta x)^2} \op \phi(x)
\nonumber \\
&\quad
    -
    \frac{2}{(\Delta x)^2}
    \sum_{n\neq 0}
    \frac{(-1)^{n}}{n^2}
    \op \phi(x + n \Delta x).
    \label{eq:BL2deriveval}
\end{align}
When $\Delta x$ generates a Nyquist lattice (or a finer one), then $\partial_x \to \partialBL_x$ and $\partial_{xx} \to \partialBL_{xx}$ are valid as exact replacements. Additionally, we define the `derivative' of the discrete field representation on the Nyquist lattice as
\begin{subequations}
\label{qpprimedefs}
\begin{align}
    \op{q}'_j
    &\coloneqq 
    \phantom{\Delta x}\,
    \partialBL \hat{\phi}(x)\Big|_{x=x_j}
    =
    \sum_{k\in\integers} D_{jk}\op{q}_k,
\\
    \op{p}'_j
    &\coloneqq
    \Delta x\,
    \partialBL \op{\pi}(x)\Big|_{x=x_j}
    =
    \sum_{k\in\integers}D_{jk}\op{p}_k.
\end{align}
\end{subequations}

We note in passing that the bandlimited first and second derivatives in Eqs.~(\ref{eq:BLderivdef}) and (\ref{eq:BL2derivdef}) are equivalent, by direct comparison, to the finite difference approximations of the first and second derivative using infinite-order stencils~ \cite{UpdatedFornberg2022}. This is discussed further in Appendices~\ref{app:BLderivs} and~\ref{app:BLsecondderivs}.

Finally, it can be shown that integration of the product of two bandlimited functions can be written as a sum of products of their samples~\cite{PhysRevLett.85.2873}:
\footnote{This is a special case of Parseval's theorem for the Fourier series~(a) and the Fourier transform~(b), under the assumption that $f$ and $g$ are bandlimited:
\begin{align*}
    \sum_{i\in\integers} [f(x_i)]^* g(x_i) \Delta x
\overset {(a)} {=}
    \int_\mathbb{R} [\tilde f(k)]^* \tilde g(k) \frac {dk} {2\pi}
\overset {(b)} {=}
    \int_\mathbb{R} [f(x)]^* g(x) dx
    .
\end{align*}
}
\begin{align}
    \sum_{i\in\integers} f(x_i)g(x_i) \Delta x = \int_\mathbb{R} f(x)g(x)dx.
    \label{eq:BLintdef}
\end{align}
In other words, a Riemann-sum approximation to the integral of a product of bandlimited functions is exact if it samples the functions at or above the Nyquist rate.

\subsection{Klein--Gordon Hamiltonian}
The free Klein--Gordon field in one dimension is a simple QFT, making it a useful starting point for determining a discrete equivalent when it is bandlimited. Its Hamiltonian is given as
\begin{align}
    \op{H} &=\frac{1}{2} \int_{\mathbb{R}} dx\, \left[ \op{\pi}^2 (x)  + [\partial_{x} \op{\phi}(x)]^2 + m^2 \op{\phi}^2(x) \right].\label{KGHam1}
\end{align}
We will now use the tools tools introduced above to determine its exact equivalent representation as a discrete system (i.e.,~without taking a continuum limit) under the assumption that $\op{\phi}$ and $\op{\pi}$ are bandlimited with a UV cutoff and can be represented using Eqs.~(\ref{phiBL}) and (\ref{piBL}). Using integration by parts to rewrite $\op{H}$ as
\begin{align}
    \op{H} &=\frac{1}{2} \int_{\mathbb{R}} dx\, \left[ \op{\pi}^2 (x)  - \op{\phi}(x) \partial_{xx} \op{\phi}(x) + m^2 \op{\phi}^2(x)\right]\label{KGHam2}
\end{align}
will make the evaluation of the Hamiltonian simpler. 

By replacing $\partial_{xx} \to \partialBL_{xx}$ and then replacing the integral with a sum of Nyquist samples using Eq.~\eqref{eq:BLintdef}, we can write
\begin{align}
    \op{H} &=\frac{\Delta x}{2}
    \sum_{j\in\integers}
    \left[ \op{\pi}^2 (x)  - \op{\phi}(x) \partialBL_{xx} \op{\phi}(x) + m^2 \op{\phi}^2(x)\right]_{x=x_j}
    .
\end{align}
We may expand the bandlimited derivative term using Eq.~\eqref{eq:BL2deriveval},
\begin{align}
&
    - \op{\phi}(x) \partialBL_{xx} \op{\phi}(x)
    \Bigr\rvert_{x=x_j}
\\*
&\quad=
    \frac{\pi^2}{3(\Delta x)^2} \op \phi^2(x_j)
    +
    \frac{2}{(\Delta x)^2}
    \sum_{n\neq 0}
    \frac{(-1)^{n}}{n^2}
    \op \phi(x_j)
    \op \phi(x_{j+n})
    ,
\nonumber
\end{align}
and then use our definitions of the field operators on a lattice, $\op{q}$ and $\op{p}$ [Eqs.~\eqref{qdef} and~\eqref{pdef}], to write
\begin{align}
    \op{H}
&=
    \frac{1}{2\Delta x}
    \sum_{j\in\integers}
    \bigg[
    \op{p}_j^2
    +
    \biggl(\frac{\pi^2}{3} + (\Delta x)^2m^2\biggr) \op{q}_j^2
\nonumber \\*
&\qquad\qquad\qquad
    +
    \sum_{n\neq0}\frac{2(-1)^n}{n^2}
    \op{q}_j\op{q}_{j+n}
    \bigg]
    .
    \label{BLKGHam}
\end{align}
Eq.~(\ref{BLKGHam}) is an exact discrete equivalent to Eqs.~(\ref{KGHam1}) and~(\ref{KGHam2}) for bandlimited $\op{\phi}$ and $\op{\pi}$ fields.

The contribution to the Hamiltonian from each $\op{q}_j$ is no longer a local coupling between field operators represented by the continuous second derivative. Instead, it is a nonlocal, weighted, alternating sum of field operators $\op{q}_j$ coupled to all others on the lattice. Figure~\ref{fig:beans} shows a representation of this. Notice that each oscillator is now coupled to all others in the lattice, and the coupling strength decays quadratically with distance.

\begin{figure}[t]
    \centering
    \includegraphics[width=\linewidth]{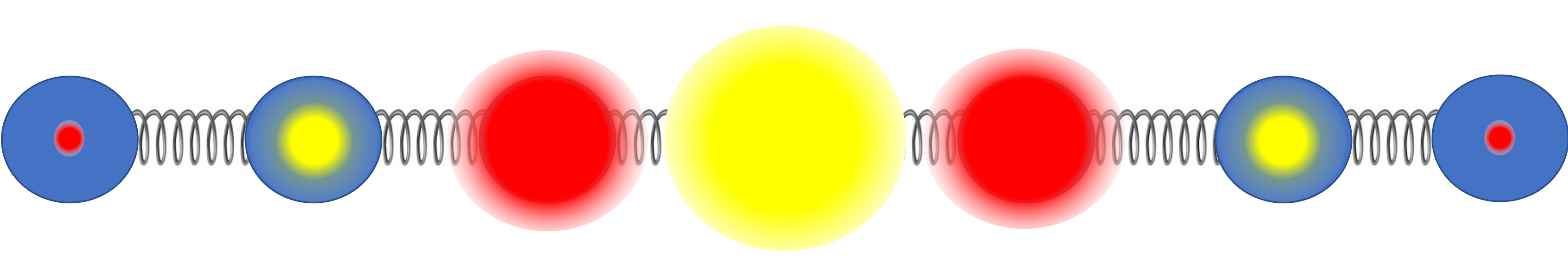}
    \caption{Graphical representation of interactions between $\op{q}$'s in a bandlimited Klein--Gordon Hamiltonian, Eq.~\eqref{BLKGHam}. A yellow glow indicates a positive contribution to the Hamiltonian, while a red glow indicates a negative contribution. The strength and sign of these contributions follows from the $(-1)^n/n^2$ term in Eq. (\ref{BLKGHam}), in that the strength of coupling between oscillators decays quadratically with distance. Note now that the presence of springs in this figure should not be taken literally as the couplings between the field values at each position are nonlocal. In fact, the couplings now reach infinitely across the field, with a strength that decays quadratically with distance between the field positions. Interestingly, it is here that we can see clearly that entirely local theories in their continuous representations can possess infinitely non-local interactions in their equivalent discrete representations. 
    }
    \label{fig:beans}
\end{figure}

It is worth considering whether the dispersion relation of this lattice model is the same as that of the continuous Klein--Gordon field. By construction, the bandlimited second derivative~$\partialBL_{xx}$ acts as an ordinary second derivative~$\partial_{xx}$ when applied to any function whose spectral support is limited to~$(-\Omega, \Omega)$. Recall that $\partial_{xx}$ is diagonal in the basis of plane waves~$e^{ikx}$, with corresponding eigenvalue~$-k^2$. This means that $\partialBL_{xx}$, in both its continuous and discrete representation, is also diagonal in this basis when restricted to only those plan waves with ${\abs{k} < \Omega}$. Thus, the Klein--Gordon Hamiltonian in its discrete representation, Eq.~\eqref{BLKGHam}, may be diagonalized in a normal-mode basis of Nyquist-sampled plane waves, and the dispersion relation will maintain the usual form, $\omega^2 = k^2 + m^2$, restricted to modes with ${\abs{k} < \Omega}$.

While this work with the Klein--Gordon field is an example of taking a bandlimited continuous QFT and representing it on a lattice with perfect equivalence, it is possible to approach this equivalence from the other direction: Starting with a Hamiltonian that is initially defined on a lattice, we can treat it as if it were the discrete representation of a continuous bandlimited field, thereby lifting the lattice model to the continuum using Shannon reconstruction, all the while avoiding ever taking the limit $\Delta x\rightarrow 0^{+}$.

\section{Discrete quantum field theories lifted to an equivalent bandlimited continuum theory \label{section:DiscreteQFT}}

In this section we show a simple example of treating a discrete field as a sample of a continuous but bandlimited field, comparing the forms of the discrete and continuous Hamiltonians with eachother as well as with those of the bandlimited Klein--Gordon field. 

\subsection{Harmonic chain Hamiltonian lifted to the continuum}
Here we will quickly derive the continuous bandlimited equivalent to the Harmonic chain Hamiltonian given by
\begin{align}
\op{H} &= \sum_{j\in\integers} \left[\frac{\op{p}_j^2}{2m} + \frac{k}{2}\left(\op{q}_{j+1}-\op{q}_j\right)^2\right].\label{HCHam}
\end{align}
Using the definitions for $\op{q}$ and $\op{p}$ in Eqs.~(\ref{qdef}) and (\ref{pdef}), Eq.~(\ref{HCHam}) can be rewritten in terms of continuous bandlimited fields using Eq.~(\ref{eq:BLintdef}):
\begin{align}
    \op{H} &= \sum_{j\in\integers} \left[\frac{\op{\pi}^2(x_j)(\Delta x)^2}{2m} + \frac{k}{2}[\op{\phi}(x_j + \Delta x)-\op{\phi}(x_j)]^2\right]\nonumber\\
    &= \frac{1}{2}\int_\mathbb{R}\left[\frac{\op{\pi}^2(x)\Delta x}{m} + \frac{k}{\Delta x}[\op{\phi}(x + \Delta x)-\op{\phi}(x)]^2\right]dx.\label{BLHCHam}
\end{align}
Notably, the nearest neighbor coupling present in the discrete representation of the harmonic chain Hamiltonian manifests itself as an interaction of the $\op{\phi}$ field with itself at a separation of $\Delta x$. Additionally, it can be seen that while only lattice translational invariance is present in Eq.~(\ref{HCHam}), full continuous translational invariance of the field is present in Eq.~(\ref{BLHCHam}). This observation will be discussed later in this section.

The continuous Hamiltonian itself can also be sampled back down to a new discrete theory using Shannon sampling. This new sampling lattice need not have oscillators at the same positions as the original discrete theory, so long as the spacing between sample points is unchanged, yet still describes the same physics as the original lattice theory. This comes from the freedom of choice of $b$ in Eq.~\eqref{xjdef} when defining the lattice of samples $x_j$. The $\op{q}$'s and $\op{p}$'s themselves of this new sampling can be expressed as linear combinations of the original oscillators through Eq.~\eqref{latticetransformation} yet will produce a model that behaves identically to the original field. As such, through its equivalence to a translationally invariant continuum model, we have shown that the discrete model for the quantum harmonic chain is entirely lattice independent. 

\subsection{Locality tradeoffs between discrete and continuous representations}
The relationship between the continuous and discrete equivalents for the bandlimited Klein--Gordon field and Harmonic chain are compared in Fig.~\ref{fig:Hamcompdiagram}.
\begin{figure}[t]
    \centering
    \includegraphics[width=\linewidth]{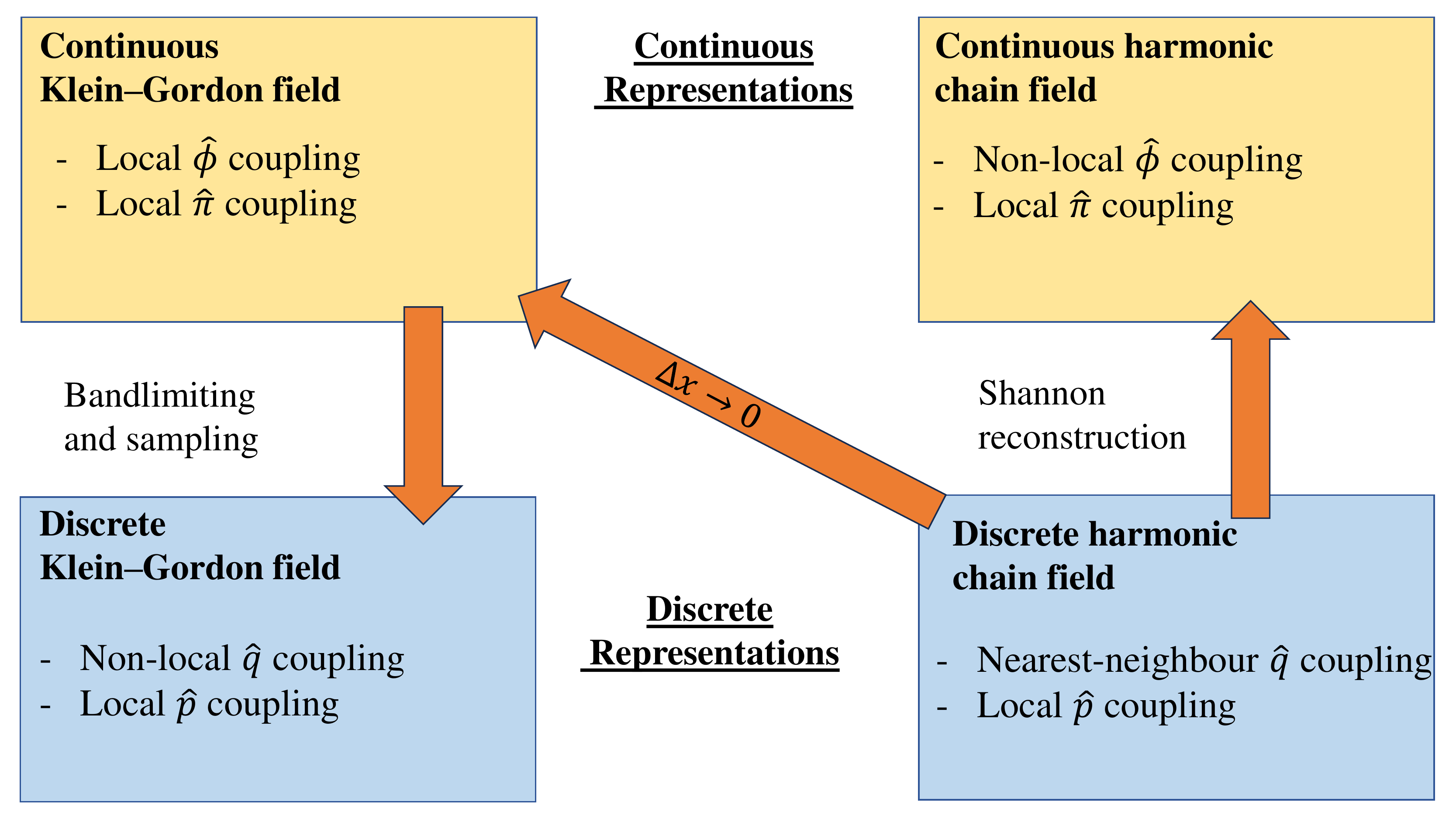}
    \caption{Diagram of relationship between the  Klein--Gordon and harmonic-chain Hamiltonians. Note that the boxes in orange are the continuous representations of these fields, while the boxes in blue are the corresponding discrete ones. The details in each box describe the type of interaction that the corresponding fields have regarding their locality. The top left, top right, bottom left, and bottom right boxes refer to Eqs. (\ref{KGHam1}), (\ref{BLHCHam}), (\ref{BLKGHam}), and (\ref{HCHam}), respectively. The orange arrows represent the action that can be taken to recover one field from another, such as taking the limit as the lattice spacing approaches zero---as represented by $\Delta x \rightarrow 0$---or using bandlimitation techniques such as sampling or Shannon reconstruction.}
    \label{fig:Hamcompdiagram}
\end{figure}

It can be seen that the different methods of producing continuous fields from the harmonic chain result in different Hamiltonians. That is, taking the limit for lattice spacing approaching zero is a very different approach to Shannon reconstruction and results in an entirely different field. This difference is mirrored in the fact that the discrete form of the bandlimited Klein--Gordon Hamiltonian in the bottom left corner of Fig.~\ref{fig:Hamcompdiagram} is very different to the Harmonic Chain Hamiltonian in the bottom right. Despite the Harmonic chain being the `traditional' discrete analog to the Klein--Gordon Field.

When comparing the discrete Hamiltonians from Fig (\ref{fig:Hamcompdiagram}), one can see that the key difference comes from the contributions of the position operators and how each $\hat{q}$ couples to others in the lattice. For the harmonic chain Hamiltonian, Eq.~\eqref{HCHam}, we can write the sum of the $\op q$ coupling terms as
\begin{align}
&
    \text{(Harmonic chain coupling terms)}
\nonumber \\*
&\quad\quad\quad
=
    \sum_{j\in\integers} 
    \frac{k}{2}(\op{q}_{j+1}- \op q_j)^2
\nonumber \\*
&\quad\quad\quad
=
    \sum_{j\in\integers}
    \frac{k}{2}
    (
    2\op q_j^2
    -
    \op q_j \op q_{j-1}
    -
    \op q_j \op q_{j+1}
    )
    ,
    \label{HCHamtransinv}
\end{align}
Graphically, the position coupling of the harmonic chain Hamiltonian can be visualized using Fig.~\ref{fig:HChainint}, where the field itself is modelled by a chain of balls and springs, with the positions of the balls denoted \(\op{q}_j\), and the interaction of each ball with its nearest neighbors contributes to the Hamiltonian of the system. We see that the interactions of each ball with its neighbors is local in the sense that no interactions exist beyond the nearest neighbors.
 
Comparatively, the position interaction terms for the bandlimited Klein--Gordon Hamiltonian are not nearest neighbor at all. Instead, they are infinitely nonlocal with alternating sign and a strength of each consecutive term 
that decays quadratically with distance. The coupling terms in the bandlimited Klein--Gordon Hamiltonian, Eq.~\eqref{BLKGHam}, are reproduced here for easy comparison with Eq.~\eqref{HCHamtransinv} above:
\begin{align}
&
    \text{(Klein--Gordon coupling terms)}
\\
&
=
    \frac{1}{2\Delta x}
    \sum_{j \in \integers}
    \left[
     \biggl(\frac{\pi^2}{3} + (\Delta x)^2m^2\right)\op q_j^2
     +
     \sum _{n\neq 0}\frac{2(-1)^n}{n^2}\op{q}_j \op{q}_{j+n}\biggl].
\nonumber
\end{align}
Here, the contribution to the Hamiltonian comes from the coupling of each \(\op{q}_j\) with every other \(\op{q}\) in the lattice. Notably, this is not just a nearest neighbor interaction, instead each point on the lattice is coupled to all others. A visualization of this interaction was shown in Fig.~\ref{fig:beans}.

\begin{figure}[t]
    \centering
    \includegraphics[width=\linewidth]{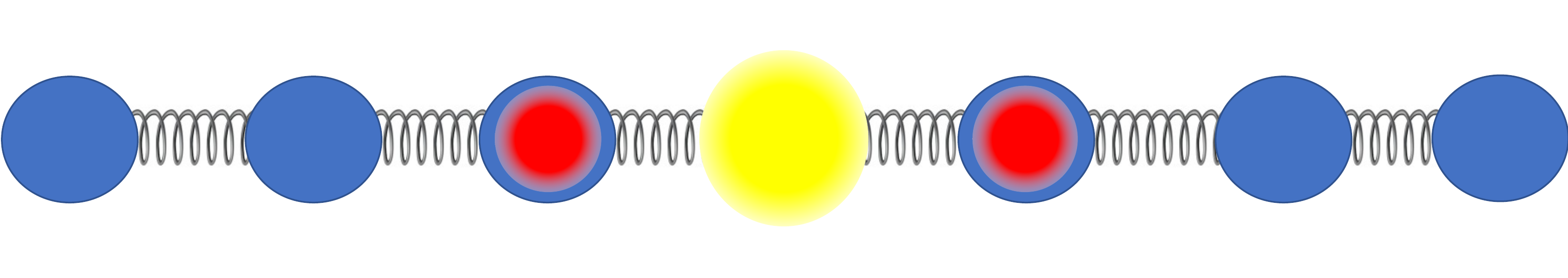}
    \caption{Graphical representation of coupling between an oscillator in a harmonic chain with its neighbours. Here, the coupling is between just \(q_j\) and its two neighbors, the sign and strength of each coupling and its contribution to the Hamiltonian is indicated by the colour and size of the glow as given by Eq. (\ref{HCHamtransinv}). A yellow glow indicates a positive contribution while a red glow indicates a negative contribution. A lack of any glow indicates no coupling between that oscillator and the center one taking place.}
    \label{fig:HChainint}
\end{figure}

 There is a stark difference between the position coupling in the bandlimited Klein--Gordon field and those in the harmonic chain. This difference comes from the fact that the two fields are discretized using completely different methods. Specifically, the harmonic chain is a field that is initially discrete, with a finite, real, length spacing between individual oscillators. The non-bandlimited Klein--Gordon field can be produced from the harmonic chain by taking the limit as the spacing between the oscillators approach zero and replacing mass with mass density. That being said, the relationship between the bandlimited Klein--Gordon field and its discrete representation is not one of taking the limit and changing the lattice spacing. Instead, it is one of restricting the density of information contained within the field \cite{UpdatedPye2015} in the sense that taking more measurements of the field than required by the Nyquist spacing does not provide any further information on the field configuration. Due to this restriction on information density, local operations on the continuous field that are more finely grained than the lattice spacing---such as the derivative---will be infinitely non-local on the lattice, as information from all other points on the lattice are required to describe continuous behavior between lattice points.

\section{Revealing continuous translation symmetry and conserved momentum within discrete models}
\label{sec:contsymm}

When comparing the continuous and discrete Hamiltonians of either the Klein--Gordon field or the harmonic chain, one can see that the continuous Hamiltonians possess fully continuous translational symmetry, while their discrete counterparts contain only a lattice symmetry at first glance. However, for bandlimited fields, the discrete is equivalent to the continuous, and thus the apparent loss (or gain) of symmetry when moving from the continuous representation to the discrete one (or the other way around) should be nothing more than an illusion. That is, these discrete fields actually contain full translational symmetry, and through Noether's theorem, there must be a conserved quantity related to said symmetry.

It is important to note here that many discrete fields have periodic boundary conditions on their Fourier transform. Changing these boundary conditions to Dirichlet boundary conditions, as is done when raising the discrete lattice to a continuum (see the discussion in Sec.~\ref{sec:bandlimitedKG}), is likely the underlying reason behind the presence of fully continuous translational invariance in these discrete fields.  For a quantum field, there must also be an operator relating to this conserved quantity that generates a translation along the symmetry. For a continuous field with full translational symmetry, this is the total momentum operator~$\op{P}$. A unitary operator in the form of an exponentiation of the total momentum operator generates a translation in the field. As such, under the assumption that a discrete field can be treated as some Shannon sampling of a continuous but bandlimited field, there must also be a total momentum operator that can generate an arbitrary \emph{continuous} translation in the discrete field---in a sense that we will clarify below. 

In order for an operator $\op{P}$ to enact horizontal translations on a field, the following conditions must be met:
\begin{subequations}
\label{eq:comconditions}
\begin{align}
    \left[\op{\phi}(x), \op{P}\right] &= 
    -i\valforhbar \partial_x \op{\phi}(x),\label{comcondition1}\\
    \left[\op{\pi}(x), \op{P}\right] &= 
    -i\valforhbar\partial_x \op{\pi}(x),\label{comcondition2}
\end{align}
\end{subequations}
where $\op{P}$ is the total momentum \cite{UpdatedPeskin2018}.
For a continuous field, the operator that satisfies these conditions is the total momentum operator. As such, the most likely candidate for an operator that satisfies Eqs.~(\ref{eq:comconditions}) for the discrete form of a bandlimited field would be the discrete form of the bandlimited total momentum. Here we introduce said operator and check if its discrete form satisfies the above conditions.

The total momentum operator for a quantum field in 1+1 dimensions is given in the form~\cite[Eq.~(2.33)]{UpdatedPeskin2018}\blk
\begin{align}
    \op{P} \coloneqq -\int_\mathbb{R}dx\, \op{\pi}(x)\partial_x\op{\phi}(x).\label{TotPdef}
\end{align}
Here, we will explore the discrete form of this operator when  $\op{\phi}$ and $\op{\pi}$ are bandlimited such that they are given by Eqs.~(\ref{phiBL}) and (\ref{piBL}), respectively. Additionally, given that the fields are bandlimited, the spatial derivative in Eq.~(\ref{TotPdef}) can also be replaced with the bandlimited spatial derivative. That is,
\begin{align}
     \op{P}
     &=
     -\int_\mathbb{R}dx\, \op{\pi}(x)\partialBL_x\op{\phi}(x),\nonumber
\\
     &=
     -\Delta x \sum_{i\in\integers} \op{\pi}(x_i)\frac{-1}{\Delta x}\sum_{n\neq0}\frac{(-e^{\Delta x \partial_x})^n}{n}\op{\phi}(x_i)\nonumber
\\
     &=
     \frac{1}{\Delta x}\sum_{i\in\integers} \sum_{n\neq 0}\frac{(-1)^n}{n}\op{p}_i \op{q}_{i+n}\nonumber
\\
    &=
     -
     \sum_{i,j \in \integers}
     D_{ij}
     \op{p}_i \op{q}_{j}
     \label{BLtotmomentum}
     ,
\end{align}
where $D_{ij}$ are elements of the derivative matrix $\mat{D}$ given in Eq.~\eqref{derivativematrixdef}. 
By using the  matrix elements $D_{ij}$ to describe $\op{P}$, we can write the bandlimited total momentum operator in terms of vectors $\opvec{p}$ and $\opvec{q}$ that are defined as column vectors with operator elements $\op{p}_j$ and $\op{q}_j$, respectively. We have
\begin{align}
    \label{matrixvectorPdef}
    \op{P} 
    =
    -\opvec{p}^\tp\mat{D}\opvec{q},
\end{align}
where $^\tp$ indicates a row vector of operators instead of a column vector. While this form of $\op{P}$ will not be used explicitly in this work, it may be of use in the future.

Now, one can check that $\op{P}$ commutes with the Hamiltonians of both the harmonic chain and the Klein--Gordon field in their lattice representations---i.e.,
\begin{align}
    \left[\op{H}_{KG},\op{P}\right] &=0\label{PCOMKG},\\
    \left[\op{H}_{HC}, \op{P}\right] &=0\label{PCOMHC}.
\end{align}
In fact, this is just a special case of the following theorem, which generalizes this commutation to all fields with quadratic Hamiltonians.

\begin{theorem}
\label{thm:Hquad}
Let $\op H_{\mathrm{quad}}$ be a discrete-translationally invariant, quadratic Hamiltonian on a one-dimensional lattice of oscillators. That is,
\begin{align}
\label{eq:Hquad}
    \op H_{\mathrm{quad}}
&
=
    \sum_{j \in \integers} \op h_j
    ,
\end{align}
where
\begin{align}
\label{eq:Hquadoneterm}
    \op h_j
&=
    \sum_{n \in \integers}
    \left(
    c^{(qq)}_n \op q_j \op q_{j+n}
    +
    c^{(qp)}_n \op q_j \op p_{j+n}
    +
    c^{(pp)}_n \op p_j \op p_{j+n}
    \right)
    .
\end{align}
In such a system, the bandlimited total momentum operator~$\op{P}$, Eq.~\eqref{BLtotmomentum}, is conserved:
\begin{align}
    [\op{H}_{\mathrm{quad}}, \op{P}]=0.
\end{align}
\end{theorem}
\begin{proof}
The proof of this theorem---and thus, by extension, Eqs.~(\ref{PCOMKG}) and~(\ref{PCOMHC})---can be found in Appendix~\ref{app:proofofHquad}.
\end{proof}

It is important to note that higher powers of the field operators in a Hamiltonian will not generally commute with $\op{P}$. One can see this immediately by checking that
\begin{align}
    \sum_{j\in\integers}\left[\op{q}_j^3, \op{P}\right]\neq 0,
\end{align}
a proof of which can also be found in Appendix~\ref{app:proofofHquad}. Thus, Hamiltonians with terms beyond quadratic order will not always%
\footnote{While interacting theories on lattices with local interaction terms such as $\op{q}^4$ do not possess continuous translational invariance. It is possible to start with a continuous interacting theory and sample this field onto a lattice using Shannon theory. This field will, by construction, possess continuous translational symmetry in its discrete representation, at the cost of being nonlocal. Given this, it may be possible to engineer non-local interacting lattice theories that are continuously translationally invariant.
We discuss this further in section \ref{sec:outlook}.}
possess continuous translation symmetry.

\subsection{Interpretation of continuous translation symmetry in a discrete system}

We are still left with the challenge of interpreting a continuous translation in a continuous system. To do this, we will use the discrete form of~$\op P$, \eqref{BLtotmomentum}, and show how to interpret the action it generates on a lattice field configuration.

Commuting $\op P$ with with the discrete representation of the field operators $\op{q}$ and $\op{p}$ at a point on the lattice gives their (bandlimited) spatial derivative:
\begin{subequations}
\label{qpcommutators}
\begin{align}
[\op{q}_j,\op{P}]
&=
    -
     \sum_{k,l \in \integers}
     D_{kl}
     \underbrace{
     [\op{q}_j,\op{p}_k]
     }
     _{i \delta_{jk}}
     \op{q}_{l}
=
    -i
    \sum_{l \in \integers}
    D_{jl}
    \op{q}_{l}
=
    -i
    \op q'_j
    ,
\\
   [\op{p}_j,\op{P}]
&=
    -
    \sum_{k,l \in \integers}
    D_{kl}
    \op{p}_k
    \underbrace{
    [\op{p}_j,\op{q}_{l}]
    }
     _{-i \delta_{jl}}
=
    -i
    \sum_{l \in \integers}
    D_{jl}
    \op{p}_l
=
    -i
    \op p'_j
    ,
\end{align}
\end{subequations}
where $\op{q}'_j$ and $\op{p}'_j$ are defined in Eq.~\eqref{qpprimedefs}.

Collecting these into operator-valued column vectors, we have
\begin{subequations}
    \begin{align}
        \left[\opvec{q}, \op{P}\right]
        &=
        -i\mat{D}\opvec{q}
        =
        -i\opvec{q}',
    \\
        \left[\opvec{p}, \op{P}\right]
        &= 
        -i\mat{D}\opvec{p}
        =
        -i\opvec{p}'.
    \end{align}
\end{subequations}
Now, we replace $\op{q}$ and $\op{p}$ (as well as $\op{q}'$ and $\op{p}'$) in Eqs.~\eqref{qpcommutators} with their definitions as samples of continuous fields $\op{\phi}$ and $\op{\pi}$ from Eq.~\eqref{qpdef} to write the commutator of $\op{P}$ with the continuous fields:
\begin{subequations}
    \label{phipicommutator}
    \begin{align}
        \left[\op{\phi}(x_j), \op{P}\right]
        &=
        -i\partialBL 
        \op{\phi}(x)
        \Bigr|_{x=x_j},
    \\
        \left[\op{\pi}(x_j), \op{P}\right]
        &=
        -i\partialBL
        \op{\pi}(x)
        \Bigr|_{x=x_j}.
    \end{align}
\end{subequations}

As $\op{P}$ is independent of choice of lattice so long as the spacing is $\frac{\pi}{\Delta x}$ or denser, Eq.~(\ref{phipicommutator}) can be extended to
\begin{subequations}
\label{eq:actionofP}
\begin{align}
    \left[\op{\phi}(x), \op{P}\right] &= -i\partialBL_x \op{\phi}(x),\\
    \left[\op{\pi}(x), \op{P}\right] &= -i\partialBL_x \op{\pi}(x).
\end{align}
\end{subequations}
As such, since we can replace $\partialBL_x \to \partial_x$ for bandlimited fields, we can see that the discrete form of the bandlimited total momentum operator does indeed satisfy the conditions for generating translations, Eqs.~\eqref{eq:comconditions}. 
Additionally, we can see that the commutator of $\op{P}$ with either the discrete or continuous representation of the bandlimited field will yield the derivative of the field. 
From Eq.~(\ref{eq:actionofP}), it follows that the exponentiation of $\op{P}$ acting on $\op{\phi}$ as a unitary operator in the Heisenberg picture generates a horizontal translation. That is, 
\begin{align}
\label{eq:HeisactionofP}
    e^{
    i
    \op{P}a}\op{\phi}(x)
    e^{
    -i
    \op{P}a} = \op{\phi}(x-a).
\end{align}
\blk
where $a$ is an arbitrary distance. When $a>0$, this action shifts the field configuration to the right by~$a$. Importantly, this arbitrary translation is possible in both the continuous and discrete representations of $\op{\phi}$, meaning that one can take discrete fields on a lattice and continuously translate them by an arbitrary distance. We define the operator that generates these translations as 
\begin{align}
    \op{U}(a) \coloneqq 
    e^{
    -i
    \op{P}a},\label{dispOP}
\end{align}
such that
\begin{align}
    \op{U}(a)^\dagger\op{\phi}(x)\op{U}(a) = \op{\phi}(x-a).
\end{align}
When considering the effect of $\op{U}$ in a continuous sense, one need only think of a usual translation along~$x$. From the discrete point of view, however, it is a little more complicated. The effect of $\op{U}$ on a discrete field on a lattice is the equivalent of doing Shannon reconstruction of the continuous field using the lattice values, doing a continuous shift on the reconstruction, and then resampling the translated field on the original  lattice. 

Such a shift has in the discrete representation has the form 
\begin{align}
    e^{
    i
    \op{P}a}\op{q}_k(x)e^{
    -i
    \op{P}a} &= \op{\phi}(x_k-a)\eqqcolon \op{q}^{(a)}_k,
\end{align}
which can be expressed in terms of the original discrete field as
\begin{align}
    \op{q}^{(a)}_k=\op{\phi}(x_k-a) = \sum_{j\in\integers}\op{q}_j\sinc_\pi\left(\frac{x_k-x_j-a}{\Delta x}\right)\label{shiftedlatticefield}.
\end{align}
This shifted discrete field~$\op{q}^{(a)}_k$ can be interpreted as a sample of the translated continuous field. Figure~\ref{fig:shifteddiscretefields} shows a graphical representation of this process.

\begin{figure}[t]
    \centering
    \includegraphics[width=\linewidth]{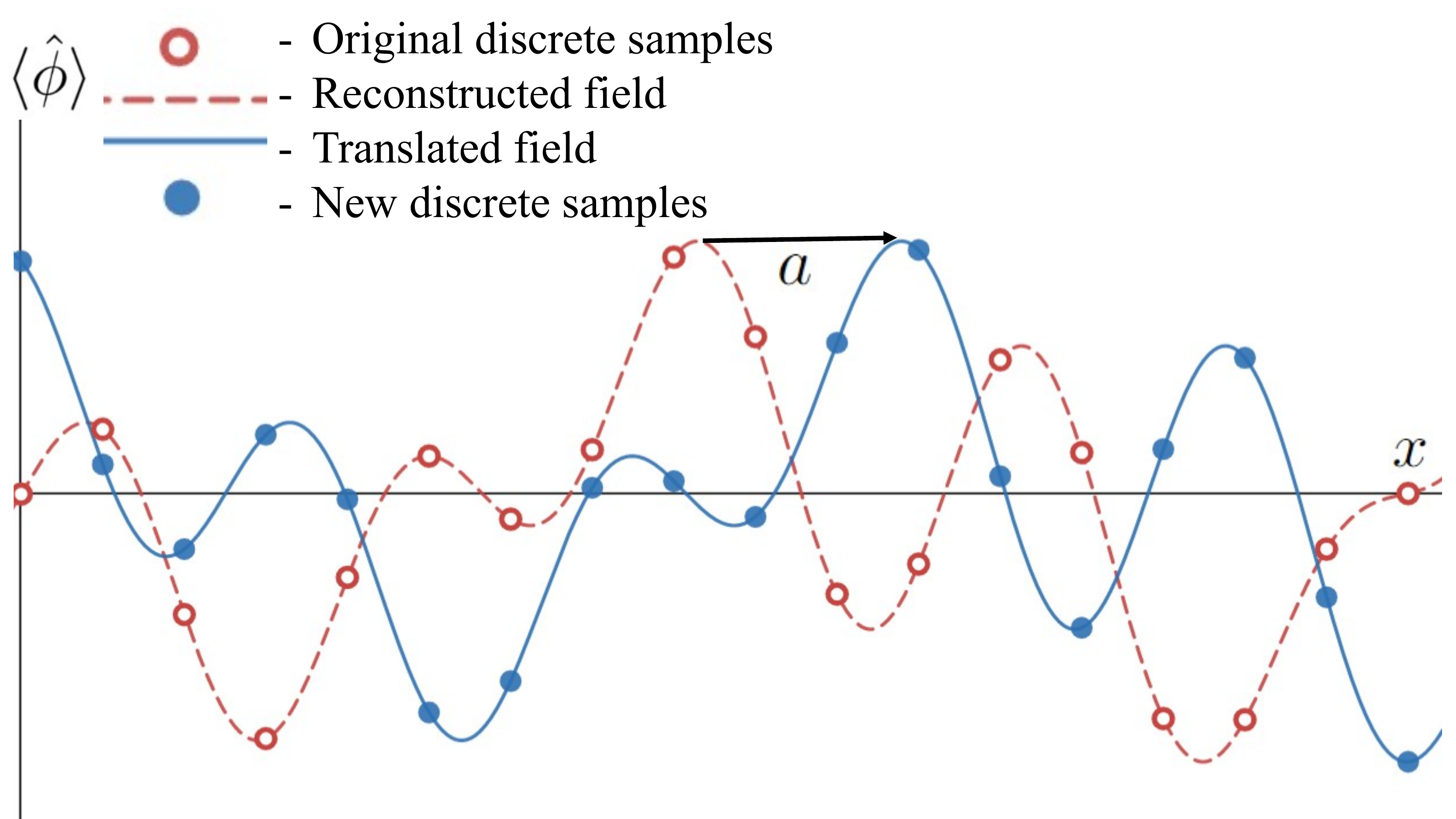}
    \caption{Graphical representation of a lattice field undergoing continuous translations. The initial lattice field operators $\op{q}_j$ (red circles) are lifted to a continuous field $\hat{\phi}(x)$ (dashed line) through Shannon reconstruction. This field is then continuously translated to the right by an arbitrary positive distance~$a$. This translated field $\hat{\phi}(x-a)$ (solid blue line) is then sampled back down onto the original Nyquist lattice  to give new samples $\hat{q}_j$ (blue dots). Note that the plots here represent the expectation value of the field at position $x$, as opposed to being plots of the field itself.}
    \label{fig:shifteddiscretefields}
\end{figure}

The samples $\op{q}^{(a)}_k$ of the translated field, taken on the original lattice, can also be used to reconstruct said continuous field using Shannon reconstruction: 
\begin{align}
    \op{\phi}^{(a)}(x) = \op{\phi}(x-a) = \sum_{k\in\integers}\op{q}^{(a)}_k\sinc_\pi\left(\frac{x-x_k}{\Delta x}\right).
\end{align}

To restate this, the displacement operator defined in Eq.~(\ref{dispOP}) allows us to shift the continuous object under the lattice of field values across by an arbitrary distance, affecting the value that the field samples take at those lattice points. This is analogous to shifting the position of a continuous field by an arbitrary distance.

Recall that $\op{P}$ commutes only with quadratic Hamiltonians on a lattice and generally fails to commute when the polynomial order of any term in the Hamiltonian is greater than two. This means that lattice QFTs described by discrete-translationally invariant Hamiltonians of order higher than quadratic do not necessarily possess continuous translational invariance. One may observe this intuitively from Eq.~(\ref{eq:BLintdef}), which shows that for sums or integrals of quadratic powers of bandlimited functions or fields, their discrete and continuous representations can be easily interchanged. 
However, Eq. \eqref{eq:BLintdef} does not apply to polynomial products of fields greater than two, meaning that while higher order polynomials of bandlimited fields still have equivalent discrete representations---by direct substitution of the Shannon sapling theorem---and vice versa, lattice translational invariance in such discrete fields cannot necessarily be promoted to continuous translational invariance.

It is still possible to start with a continuous translationally invariant interacting field and sample down to a lattice field that possesses this continuous translational invariance, by virtue the the lattice counterpart still being equivalent to a continuous theory. 

To recap, we have shown that a discrete-translationally invariant \emph{discrete} field with only quadratic terms in its Hamiltonian possesses continuous translational symmetry when it is treated as a sampling of a continuous, bandlimited field.

\section{ Discussion and outlook \label{sec:outlook}}
Throughout this paper, it has been evident that the transition from a bandlimited continuous theory to a discrete one (or vice versa) affects a few concepts of quantum field theory in interesting ways. Two key examples are locality (particularly the locality of the derivative and the quasi-locality of nearest neighbor interactions) and translational symmetry.

For the case of the bandlimited derivative when treated in the continuous sense, the derivative is exactly equivalent to the standard derivative when it is acting on a function that is bandlimited with the appropriate UV cutoff. However, when working with the discrete form of the bandlimited derivative acting on either a bandlimited continuous or discrete function or field, at first glance, the derivative appears to act very differently to its continuous counterpart due to the appearance of an infinite sum over all values on the lattice (analogous to an infinitely distributed weighted average). Yet, these two approaches must be exactly equivalent so long as their UV cutoffs are compatible. That is, the discrete form of the bandlimited derivative as an infinite weighted sum can act on the continuous form of a bandlimited function or field and produce the same result as one would get from taking the derivative of said function normally.

The bandlimited derivative itself, while determined through an information theoretic approach, is equivalent in form to the SLAC derivative used in lattice QCD \cite{UpdatedDrell1974} and the infinite order stencil approach to finite difference derivative approximations \cite{UpdatedFornberg2022}. However, from the information theoretic approach we can see that the bandlimited derivative (and by extension its SLAC and stencil equivalents) is entirely lattice independent (through the free choice of $b$ in Eq.~\eqref{xjdef} and can act on either continuous or discrete functions provided the samples of said functions are sufficiently dense. Additionally, we see that when acting on functions that do possess an inherent ultraviolet cut-off, the SLAC and stencil derivatives cease to be approximations.

If the discrete bandlimited derivative acts on a bandlimited field with UV cutoff less than the one on the derivative---i.e.,~the lattice spacing used for the discrete derivative is smaller than the Nyquist spacing of the bandlimited field, then the derivative will be oversampled and still be equivalent to the standard continuous derivative. However, if the lattice spacing for the bandlimited derivative is greater than the Nyquist spacing or if the bandlimited derivative is acting on a field that is not bandlimited, then the bandlimited derivative will only give an approximation of the standard derivative, even if infinite samples are taken. That being said, it would be interesting to investigate the precision with which the derivatives of bandlimited functions and fields can be approximated using Shannon reconstruction and the bandlimited derivative using only a finite number of samples.

It is particularly interesting that several different approaches to determining a discrete representation of the derivative have resulted in the same form. The fact that this derivative can act on both discrete and continuous functions, maintains lattice independence, and is an exact equivalent to the derivative when acting on bandlimited functions is an indication of the value that this derivative may have as an analytical and computational tool that is more powerful than was previously believed. Additionally, we show in appendix \ref{app:BLsecondderivs} that compositing the bandlimited derivative map twice produces the second bandlimited derivative map. This is an indication that there may be more to the notion of bandlimited calculus than just an approximation (or discrete equivalent) to the derivative. Work is being done to generalise this composition to higher order derivatives and to other rules of derivative calculus such as the product rule. The toolbox of bandlimited calculus may well acquire more and more discrete equivalents to the conventional methods of differential and integral calculus.

Furthermore, the continuous-discrete equivalence of the bandlimited derivative---along with the continuous-discrete equivalences of the studied quantum field theories---may also be connected to perfect lattice actions studied in the 1990s \cite{BIETENHOLZ1997, Hasenfratz1994}, where non-local lattice operations are used to approximate or replicate continuous physics results in quantum field theory and high energy physics. It is an open question as to what precisely the connection is between these perfect lattice actions and bandlimitation and how these two approaches to a continuous-discrete equivalence relate. Such an investigation is beyond the scope of this work but is an interesting open question that is worth further study in followup work.

For the case of the restoration of translational symmetry when moving from a discrete field to a bandlimited continuous one, one must ask the question: Do discrete fields always have full translational symmetry that we have not previously observed, or does treating a discrete field as a sampling from a continuous but bandlimited one generate this translational symmetry? This is an interesting question that leaves room for some interpretation. 

While it is always possible to consider a lattice field as a set of samples of some continuous but bandlimited field, it is not necessarily true that every such continuous field is translationally invariant, even if the field is translationally invariant on a lattice. For instance, we have shown that momentum is not always conserved for local interacting fields on a lattice and, as such, their associated continuous fields will not possess continuous translational invariance. 
An interesting place to explore further details on this may come from the work of Pye~\cite{UpdatedPye2015} and Grimmer~\cite{UpdatedGrimmer2022} and the notion of the inability to exactly localize a lattice point as a sample from a bandlimited field. With that being said, it would be interesting to investigate the effects that this notion of full translational symmetry on lattices would have on gauge theories as well as interacting field theories.

It will be an interesting challenge to try to define an interacting field on a lattice that \emph{does} possess fully continuous translational invariance. However, such interactions would likely be highly nonlocal on the lattice. We know that interacting lattice theories with continuous translation symmetry \emph{can} exist as we can start with a continuous interacting theory and use Shannon sampling to produce its equivalent lattice theory. As such, through the study of the forms of such theories on a lattice, it may be possible to engineer an interacting lattice theory that generally possesses continuous translation symmetry. Ongoing work is being done in this regard to further improve lattice models that replicate continuous theories with interactions. 
We also note that while quadratic fields in flat spacetime do not interact, it is possible for such fields to interact with gravity in curved spacetime \cite{UpdatedHollands2015, UpdatedParker2009}. As such, it would be an interesting avenue for future work to investigate the possible continuous translation symmetry of quadratic discrete field on curved spaces, while acknowledging the limitations of this approach since our bandlimitation is non-covariant.

The bandlimited total momentum operator, being a generator of arbitrary translations of a lattice field theory, may have applications in quantum information in determining the group velocity and hitting time of quantum random walks, extending upon work by Kempf and Portugal~\cite{PhysRevA.79.052317}. Extending the notion of continuous symmetry to discrete fields in higher than $1+1$ dimensions is also an interesting avenue to explore. It seems unlikely that restoring continuous rotational symmetry to a discrete field would work similarly to the case of translational symmetry, but it would be interesting to explore whether it can be done.

\begin{acknowledgements}
We thank Jason Pye, Nicholas Funai, Daniel Grimmer, and Julian Greentree for their support and expertise on this work. This work was supported by the Australian Research Council Centre of Excellence for Quantum Computation and Communication Technology (Project No.~CE170100012), the Australian Research Council Discovery Program (Project No.~DP200102152), and a Discovery Grant from the National Science and Engineering Research Council (NSERC) of Canada. 
\end{acknowledgements}

\appendix
\section{Simple proof of the Shannon sampling theorem}\label{app:shannonproof}
This is a simplified adaptation of Shannon's sampling theorem~\cite{UpdatedShannon1948} with an included proof following along the lines of Ref.~\cite{1455040}.  \blk
\begin{theorem}
Let $f$ be a continuous bandlimited function $f : \reals \to \complex$ with finite, open spectral support $(-\Omega, \Omega)$, where $\Omega$ is the bandlimit, i.e.,~the UV cutoff. Then, $f$~can be perfectly reconstructed from a lattice of samples using \blk
\begin{align}
    f(x) = \sum_{j\in\integers}f(x_j) \sinc_\pi\left(\frac{x-x_j}{\Delta x}\right),
\end{align}
where $\Delta x = \frac{\pi}{\Omega}$.
\end{theorem}
\begin{proof}

We start with a bandlimited function $f$ with a UV cutoff $\Omega$: 
\begin{align}
 f(x) = \int_{-\Omega}^\Omega \frac{dk}{2\pi}\tilde{f}(k)e^{ikx}.
\end{align}
We can rewrite $\tilde{f}(k)$ using its Fourier series.  
\begin{align}
 \tilde{f}(k) &= \sum_{j\in\integers} c_j e^{\frac{2\pi i jk}{2\Omega}} \label{fourierseries1}\\
 c_j &= \frac{1}{2\Omega}\int_{-\Omega}^\Omega dk \tilde{f}(k)e^{-\frac{2\pi i jk}{2\Omega}}. \label{coefficient1}
\end{align}
Now, noting that $\frac{\pi}{\Omega} = \Delta x$, we can rewrite Eq.~(\ref{coefficient1}):
\begin{align}
 c_j &= \frac{1}{2\Omega} \int_{-\Omega}^\Omega dk \tilde{f}(k) e^{-ikj\Delta x}\nonumber\\
 &= \frac{\pi}{\Omega}f(j\Delta x). 
\end{align}
Combining these equations, we can write $f(x)$ as
\begin{align}
 f(x) &= \frac{\pi}{\Omega}\int_{-\Omega}^\Omega \frac{dk}{2\pi}\sum_{j\in\integers} f(j\Delta x)e^{ik(x- j\Delta x) }\nonumber\\
 &= \frac{\pi}{\Omega}\sum_{j\in\integers} f(j\Delta x)\int_{-\Omega}^\Omega \frac{dk}{2\pi}e^{ik(x-j\Delta x)}\nonumber\\
 &= \frac{\pi}{\Omega} \sum_{j\in\integers} f(j\Delta x) \frac{\sin[\Omega (x-j\Delta x)]}{\pi(x-j\Delta x)}\nonumber\\
 &= \sum_{j\in\integers} f(j\Delta x) \sinc_\pi \left(\frac{x-j\Delta x}{\Delta x}\right).\label{shannonsamplingproof}
\end{align}
Now, noting that $j\Delta x$ is simply the $j^{th}$ point on the lattice $x_j$, we can rewrite Eq.~(\ref{shannonsamplingproof}) in the form used throughout this work:
\begin{align}
 f(x) = \sum_{j\in\integers} f(x_j)\sinc_\pi\left(\frac{x-x_j}{\Delta x}\right).
\end{align}
\end{proof}

\section{First derivatives of bandlimited functions}\label{app:BLderivs}

A bandlimited function can be equivalently described as  both continuous and discrete, and as such it follows that many continuous operations that can be done on bandlimited functions must also have discrete equivalents. 

Here we will introduce the discrete form of the derivative of bandlimited functions. Note that we simply call this the \emph{bandlimited derivative} throughout this work for brevity. \blk

The derivative of a bandlimited function can be found simply by taking the derivative of the Shannon reconstruction formula:%
\footnote{We note for completeness that the sinc function is closely related to the spherical Bessel functions~$j_n(x)$. In particular, $\sin(x) = j_0(x)$, and derivatives of sinc are related to those of the $j_n$s, although not in an overly simplistic way. In fact, $j_n(x) = (-x)^n (x^{-1} \partial_x)^n \sinc(x)$, so that $\sinc(x) = j_0(x)$, $\sinc'(x) = -j_1(x)$, but the simple pattern then breaks with $\sinc''(x) = j_2(x) - x^{-1} j_1(x)$. However, for some applications this relationship may be helpful, so we include it for reference. We choose not to use it here because it does not assist in simplifying the expressions to any meaningful extent that we can see.}
\begin{align}
\label{eq:pdfx}
    \partial_xf(x) &= \sum_{k\in\integers} f(x_k)  \partial_x \blk \sinc_\pi \left(\frac{x-x_k}{\Delta x}\right),\nonumber\\
    &=  \sum_{k\in\integers} f(x_k) \frac{1}{x-x_k}
    \nonumber \\
    &\quad \times
    \left[\cos \left(\frac{\pi}{\Delta x}(x-x_k)\right)-\sinc_\pi \left(\frac{x-x_k}{\Delta x}\right)\right].
\end{align}
When evaluated at a point $x_j$ on the same lattice as $x_k$ this reduces to
\begin{align}
    \partial_xf(x)\Bigg|_{x=x_j} 
    &= \sum_{k\in\integers} f(x_k) \frac{1}{\Delta x} \frac{(-1)^{j-k} - \delta_{jk}}{j-k}.\label{BLderivative}
\end{align}
The fraction in Eq.~(\ref{BLderivative}) can be shown to equal zero when $i=j$ by treating $x_k$ and $x_j$ as real numbers in Eq.~\eqref{eq:pdfx} and taking the limit of $x_k \rightarrow x_j$. As such, we can write the bandlimited derivative in the form
\begin{align}
    f'(x_j) = \partial_xf(x)\Bigg|_{x=x_j} &= \sum_{k\in\integers} {D_{jk}}f(x_k), \label{BLderivative2}
\end{align}
where ${ D_{jk}}$ is given by 
\begin{align}
    {D_{jk}} 
    &=
        \begin{cases}
        \displaystyle
        \frac{1}{\Delta x} \frac{(-1)^{j-k}}{j-k} &\quad \text{if $j\neq k$},\\
        \displaystyle\vphantom{\frac{0}{0}}0 &\quad \text{if $j=k$}, 
        \end{cases}
 \label{eq:Dmatrix}
\end{align}
as defined in Eq.~\eqref{derivativematrixdef}. 

 Since $D_{jk}$ are elements of a Toeplitz matrix, as mentioned in section \ref{sec:bandlimitedKG}, 
 we can offer some additional intuition about the bandlimited derivative as follows. By rewriting Eq.~(\ref{BLderivative}) and reindexing the sum such that $k-j = m$, we have
\begin{align}
    \partial_xf(x)\bigg|_{x=x_j} &= \frac{-1}{\Delta x} \sum_{m\neq0}\frac{(-1)^m}{m}f(x_{j+m}),\nonumber\\
    &= \frac{-1}{\Delta x} \sum_{m\neq0}\frac{(-1)^m}{m}f(x_j + m\Delta x), \label{BLderivative3}
\end{align}
which agrees with Eq.~\eqref{eq:BLderiveval}.
Hence, the bandlimited derivative can also be viewed as a weighted sum of all points on the lattice except the point at which the bandlimited derivative is being taken. The fall-off in contribution is proportional to $m^{-1}$, where $m$ is the number of lattice steps away from the point in question.

It is interesting to consider this interpreation in contrast with the usual, finitely extended, discrete approximations to continuous derivates. The forward difference is simply $(\Delta x)^{-1}[f(x_{j+1}) - f(x_j)]$, the backwawrd difference is similarly, $(\Delta x)^{-1}[f(x_{j}) - f(x_{j-1})]$, yet the average of these gives a better estimate, the centered difference $(2\Delta x)^{-1}[f(x_{j+1}) - f(x_{j-1})]$. In fact, the bandlimited derivative results from the infinite limit of higher orders of finite-difference approximations \cite{UpdatedFornberg2022}.

 Using Eq.~\eqref{eq:transop}, we can express Eq.~\eqref{BLderivative3} using the translation operator~$e^{a\partial_x}$: 
\begin{align}
    \partial_xf(x)\bigg|_{x=x_j} &= \frac{-1}{\Delta x}\sum_{m\neq0}\frac{(-1)^m}{m}e^{m\Delta x \partial_x}f(x)\bigg|_{x=x_j},\nonumber\\
    &= \frac{-1}{\Delta x}\sum_{m\neq0}\frac{(-e^{\Delta x \partial_x})^m}{m}f(x)\bigg|_{x=x_j}.
\end{align}
Since the choice of lattice is arbitrary, this holds for any point~$x$, so we can write
\begin{align}
    \partial_xf(x) &= \frac{-1}{\Delta x}\sum_{m\neq0}\frac{(-e^{\Delta x \partial_x})^m}{m}f(x)
\end{align}
since $f$ is bandlimited. Now we can define the bandlimited partial 
derivative clearly as an infinite sum of horizontal transformations of spacing $\Delta x$:
\begin{align}
\label{eq:pBLx}
    \partialBL_x \coloneqq \frac{-1}{\Delta x}\sum_{m\neq0}\frac{(-e^{\Delta x \partial_x})^m}{m},
\end{align}
which agrees with Eq.~\eqref{eq:BLderivdef}.

To gain further insight into this form, \blk we consider the following equivalence:
\begin{align}
\label{eq:partialxaslog}
    \partial_x
&=
    \frac{-1}{\Delta x}
    (-\Delta x \partial_x)
=
    \frac{-1}{\Delta x}
    \log (e^{-\Delta x \partial_x}).
\end{align}
The only issue with this is that the log is multivalued. We choose the branch cut such that $\Im \log (\cdot) \in (-\pi, \pi)$. Since the derivative operator has plane-wave eigenfunctions~$e^{i k x}$, with associated purely imaginary eigenvalues~$i k$, the eigenvalues of $e^{-\Delta x \partial_x}$ are $e^{- i k \Delta x } = e^{- i \pi k / \Omega }$. Given the branch we chose, taking the log of this formula only returns $-i \pi k / \Omega$ if $\abs k < \Omega$. Here we see bandlimitation arising for a functional-analytic reason. The final path to obtain Eq.~\eqref{eq:pBLx} is to first use the relation
\begin{align}
    \log (a)
&=
    \log \left(\frac{1+a}{1+a^{-1}} \right)
\nonumber \\
&=
    \log (1+a) - \log (1+a^{-1})
\end{align}
and expand the logarithms in a Taylor series about~1. Let us consider the first case, which gives
\begin{align}
    \log (a)
&=
    -
    \sum_{m=1}^\infty \frac {(-1)^m} {m} a^m
    +
    \sum_{m=1}^\infty \frac {(-1)^m} {m} a^{-m}
    .
\end{align}
The first series converges for $\abs a < 1$ by analyticity, and this convergence extends to $\abs a \leq 1$ excluding $a = -1$. Analogously, the second series converges for $\abs a \geq 1$ excluding $a=-1$. So the two series together converge only on the unit circle for~$a$ except at $a=-1$. We can change summation variables of the first series to combine it with the second to obtain
\begin{align}
&    \log (a)
=
    \sum_{m\neq 0} \frac {(-1)^m} {m} a^{-m},
\\
\nonumber
& \forall a \in \{z : \abs z = 1, z \neq -1\}
    .
\end{align}
In the case we want to use this formula, $a$ is being replaced with $e^{- \Delta x \partial_x}$, which has eigenvalues~$e^{- i \pi k / \Omega}$ since $\Delta x = \pi / \Omega$. Convergence requires that $-\pi < - \pi k / \Omega < \pi$, which reduces exactly to our bandlimit condition~$\abs k < \Omega$ ---but this time, for a functional-analytic reason. Finally, plugging this into Eq.~\eqref{eq:partialxaslog} gives an expression for $\partialBL_x$:
\begin{align}
\label{eq:partialxexpanded}
    \partial_x
    \sim
    \partialBL_x
&\coloneqq
    \frac{-1}{\Delta x}
    \sum_{m\neq 0} \frac {(-1)^m} {m}
    e^{m \Delta x \partial_x},
    .
\end{align}
where $\sim$ represents the fact that these two operators are exactly equal when the function $f$ being acted upon is bandlimited as above. Rigorously, then,
\begin{align}
    \partial_x f(x) &= \partialBL_x f(x)
\end{align}
if $f$ is bandlimited. Since $\partialBL_x$ is defined with respect to a spacing~$\Delta x$, this is what determines the bandlimit for $f$ for which the ordinary derivative is recovered.

\section{Second derivatives of bandlimited functions}\label{app:BLsecondderivs}

Here we will quickly show a derivation of the second derivative by taking the second derivative of the Shannon reconstruction formula. 
\begin{align}
\label{eq:secondderivfinit}
    {\partial_{xx}} f(x) &= 
    \sum_{k\in\integers} f(x_k)
    {\partial_{xx}}
    \sinc_\pi \left(\frac{x-x_k}{\Delta x}\right)
\\
    &= \sum_{k\in\integers} f(x_k)\frac{\pi}{\Delta x}\bigg[-\frac{2\Delta x \cos \left(\frac{\pi}{\Delta x}(x-x_k)\right)}{\pi(x-x_k)^2}
\nonumber\\*
    &\quad-\frac{\pi}{\Delta x}\sinc_\pi\left(\frac{x-x_k}{\Delta x}\right)+\frac{2\Delta x\sinc_\pi\left(\frac{x-x_k}{\Delta x}\right)}{\pi(x-x_k)^2}\bigg].
\nonumber
\end{align}
Evaluating this at a point on the lattice $x=x_j$ creates a removable singularity in the $k=j$ term in the sum, so we will take the limit $x \to x_j$ to evaluate that term:
\begin{align}
    \lim_{x\rightarrow x_j}
    \partial_{xx} \sinc_\pi \left(\frac{x-x_j}{\Delta x}\right)f(x_j)
&=-\frac{\pi^2}{3 (\Delta x)^2}f(x_j),
\end{align}
which follows from the Taylor expansion
\begin{align}
    \sinc_\pi x
&=
    \sum_{n=0}^\infty
    \frac{(-1)^n (\pi x)^{2n}}{(2n+1)!}
=
    1 - \frac{\pi^2 x^2}{3!} + O(x^4).
\end{align}
For all other lattice points~$x_j$, $j\neq k$, only the first term in the summand of Eq.~\eqref{eq:secondderivfinit} is nonzero.
Thus, the second derivative evaluated at $x = x_j$ can be written as
\begin{align}
&
      {\partial_{xx}} f(x)\bigg\rvert_{x=x_j}
\\*
&\quad = -\frac{\pi^2}{3(\Delta x)^2}f(x_j)-\frac{2}{(\Delta x)^2}\sum_{k\neq j}f(x_k)\frac{(-1)^{j-k}}{(j-k)^2}
,
\nonumber \\
&\quad = -\frac{\pi^2}{3(\Delta x)^2}f(x_j)-\frac{2}{(\Delta x)^2}\sum_{m\neq 0}f(x_{j+m})\frac{(-1)^{m}}{m^2}
\nonumber
,
\end{align}
which agrees with Eq.~\eqref{eq:BL2deriveval}. As with the first derivative, we can write this using the displacement operator:
\begin{align}
&
    {\partial_{xx}} f(x)\bigg\rvert_{x=x_j}
\\*
&\quad= 
     -\frac{\pi^2}{3(\Delta x)^2}f(x_j)-\frac{2}{(\Delta x)^2}\sum_{m\neq 0}\frac{(-e^{\Delta x\partial_x})^{m}}{m^2}f(x_j),
     \nonumber
\end{align}
Since the lattice we use is arbitrary, we can define the bandlimited second derivative operator as
\begin{align}
    {\partialBL_{xx}} 
    &\coloneqq
    -\frac{\pi^2}{3(\Delta x)^2}-\frac{2}{(\Delta x)^2}\sum_{m\neq 0}\frac{(-e^{\Delta x\partial_x})^{m}}{m^2}\label{appblderiv2}
    ,
\end{align}
which agrees with Eq.~\eqref{eq:BL2derivdef}.
For another perspective on this operator, we can start with Eq.~(\ref{appblderiv2}) and recover the second derivative using the polylogarithm function~\cite{goncharov1995polylogarithms}, defined using the power series\blk
\begin{align}
\Li_s (z) &= \sum_{n = 1}^\infty \frac{z^n}{n^s}, \hspace{0.5cm}  \abs z < 1. \label{polylogdef}
\end{align}
This function can be extended by analytic continuation to all $z \in \complex$, with a branch point at $z=1$ and a branch cut typically chosen along the real axis, $(1, \infty)$. The special case $s=2$ is called the dilogarithm~$\Li_2(z)$, for which the following identity will turn out to be useful:
\begin{align}
\Li_2 (z) + \Li_2 (z^{-1}) &= -\frac{1}{2}\left(\frac{\pi^2}{3} + [\ln{(-z)}]^2\right)\label{dilogID},
\\ &\hspace{2cm} z\in \complex \setminus(1,\infty).\nonumber
\end{align} 
The sum over $m$ in Eq.~(\ref{appblderiv2}) can be split into two dilogarithm functions that have the same form as the left-hand side of Eq.~(\ref{dilogID}), giving
\begin{align}
    \partialBL_{xx}
&=
    \frac{-1}{(\Delta x)^2}
    \Bigg[\frac{\pi^2}{3}+
    2
    \sum_{n\neq0} \frac{(-e^{\Delta x \partial_x})^m}{m^2}\Bigg]
\nonumber\\ &=
    \frac{-1}{(\Delta x)^2}\Bigg[\frac{\pi^2}{3}+2\Li_2(-e^{\Delta x \partial_x})
    +2\Li_2(-e^{-\Delta x \partial_x})\Bigg]\nonumber\\
    &={\frac{1}{(\Delta x)^2}}\left[\ln(e^{\Delta x \partial_x})\right]^2\nonumber\\
    &= 
    \partial_{xx} \blk
    \label{dilogderivative}.
\end{align}
 
The UV cutoff---while not explicitly present in the final iteration of Eq.~(\ref{dilogderivative})---can also be recovered from the limiting condition of Eq.~(\ref{dilogID}). The condition on $z$ in Eq.~(\ref{dilogID}) informs us of the condition on $-e^{\Delta x \partial_x}$ and by extension, the limiting condition on the spectrum of the derivative operator. 
The eigenvalues of the derivative operator are imaginary. Thus, the spectrum of the argument of the dilogarithm is a complex phase that, from the limiting condition on Eq.~(\ref{dilogID}), cannot cross the positive real axis. Accounting for the negative sign in this argument, we can write 

\begin{align}
     \operatorname{Spec} (\Im(\Delta x\partial_x)) \in 
(-\pi,\pi)
,
\end{align}
 where Spec indicates the eigenvalue spectrum of the operator. 
Since \(\Delta x = \frac{\pi}{\Omega}\), and assuming $\operatorname{Spec} (\Delta x\partial_x)$ \blk goes around the complex plane only once, the limiting condition can be simply written as 
\begin{align}
    \operatorname{Spec}( \Im \partial_x) \in (-\Omega, \Omega).
\end{align} 
As such, Eq.~(\ref{dilogderivative}) is applicable only when the spectrum of the derivative operator is bounded by the UV cut-off. In other words, Eq.~(\ref{dilogderivative}) is true only when the function or field that the derivative operator is acting upon is bandlimited with a UV cut-off.

Similar to the case of the first derivative, we write a matrix $\mat D_{(2)}$ that acts on a vector of function values on the lattice $\bold{f}$ and maps them to a vector of the values of the function's second derivative on the lattice $\bold{f''}$, in that
\begin{align}
    \bold{f''} = \mat D_{(2)}\bold{f},
\end{align}
where 
\begin{align}
    {\partial_{xx}} f(x_j)=\partial_x\partial_x f(x)\bigg|_{x=x_j}=\sum_{i\in \integers} {[D_{(2)}]}_{jk}f(x_k),
\end{align}
and ${[D_{(2)}]}_{jk}$ is given by
\begin{align}
\label{eq:D2jk}
     {[D_{(2)}]}_{jk}\coloneqq
    \begin{cases}
        -\frac{\pi^2}{3(\Delta x)^2},&j=k\\
        -\frac{2}{(\Delta x)^2}\frac{(-1)^{(j-k)}}{(j-k)^2},&j\neq k,
    \end{cases}
\end{align}
as seen in section \ref{sec:bandlimitedKG}.
Note that these values of ${[D_{(2)}]}_{jk}$ agree with the infinite-order-stencil finite-difference approximation to the second derivative of a function \cite{UpdatedFornberg2022}.
We can verify that $\mat D_{(2)}=\mat D^2$ by checking $ {[D_{(2)}]}_{jk} = \sum_l D_{jl}D_{lk}$. Evaluating the right-hand side using Eq.~\eqref{eq:Dmatrix} gives
\begin{align}
\label{eq:D2asproduct}
    \sum_{l\in\integers}D_{jl}D_{lk}
    &= 
    \frac{1}{(\Delta x)^2}
    \sum_{l \not\in \{j,k\}}
    \frac{(-1)^{l-j}}{(l-j)}
    \frac{(-1)^{k-l}}{(k-l)}
\nonumber \\
    &= 
    \frac{-(-1)^{k-j}}{(\Delta x)^2}
    \underbrace
    {
    \sum_{l \not\in \{j,k\}}\frac{1}{(l-j)(l-k)}
    }
    _{\displaystyle S}.
\end{align}

We will now evaluate the sum~$S$ noted above. Assuming $j \neq k$, and writing $m = k-j$, we can reindex the sum in two different ways ($l \mapsto l \pm m$) and take an average:
\begin{align}
    S
&=
    \frac 1 2
    \sum_{l \not\in \{0,m\}}\frac{1}{l(l-m)}
    +
    \frac 1 2
    \sum_{l \not\in \{0,-m\}}\frac{1}{(l+m)l}
    .
\nonumber
\\
&=
    \frac{1}{2m^2}
    +
    \sum_{l \not\in \{0, \pm m\}}
    \frac{1}{2l}
    \left[
    \frac{1}{l-m}
    +
    \frac{1}{l+m}
    \right]
\nonumber \\
&=
    \frac{1}{2m^2}
    +
    \sum_{l \not\in \{0, \pm m\}}
    \frac{1}{l^2-m^2}
,
\end{align}
where, in the second line, we separated out the $l=-m$ in the first sum and the $l=m$ term in the second.
Now consider the absolutely convergent series expansion~\cite{UpdatedRemmert1991}, with $z \in \complex \setminus \integers$,
\begin{align}
    \frac{\pi \cot \pi z}{z}
&=
    \frac{1}{z^2}
    +
    \sum_{l \neq 0}
    \frac{1}{z^2-l^2}
    .
\end{align}
This function has poles for all $z \in \integers$. The pole at $z=0$ is already isolated into a separate term (the first one), and we can separate out the ones at $z=\pm m$ similarly:
\begin{align}
    \frac{\pi \cot \pi z}{z}
&=
    \frac{1}{z^2}
    +
    \frac{2}{z^2-m^2}
    +
    \sum_{l \not\in \{0,\pm m\}}
    \frac{1}{z^2-l^2}
    .
\end{align}
Rearranging and taking the limit $z \to m$, along with the Laurent expansion
\begin{align}
    \frac{\pi \cot \pi z}{z}
&=
    \frac{1}{m(z-m)}
    -
    \frac{1}{m^2}
    +
    O(z-m),
\end{align}
lets us evaluate
\begin{align}
    S
&=
    \frac{1}{2m^2}
    +
   \lim_{z \to m}
    \sum_{l \not\in \{0,\pm m\}}
    \frac{1}{l^2-z^2}
\nonumber \\
&=
    \frac{1}{2m^2}
    +
    \lim_{z \to m}
    \left[
    \frac{1}{z^2}
    +
    \frac{2}{z^2-m^2}
    -
    \frac{\pi \cot \pi z}{z}
    \right]
\nonumber \\
&=
    \frac{1}{2m^2}
    +
    \lim_{z \to m}
    \left[
    \frac{2}{m^2}
    -
    \frac{1}{m(z+m)}
    \right]
\nonumber \\
&=
    \frac{2}{m^2}
    .
\label{cotformula}
\end{align}

Substituting this into Eq. (\ref{eq:D2asproduct}) gives
\begin{align}
    \sum_{l\in\integers}
    D_{jl}D_{lk} 
    &=
    \frac{-2(-1)^{j-k}}{(\Delta x)^2 (j-k)^2}
&
    (j \neq k)
    .
    \label{eq:D2jneqkfinal}
\end{align}
When $j=k$, we can directly evaluate
\begin{align}
    \sum_{l\in\integers}D_{jl}D_{lj}
=
    \frac{-1}{(\Delta x)^2}
    \sum_{l\neq j}
    \frac{1}{(l-j)^2}
=
    \frac{-1}{(\Delta x)^2}
    \frac{\pi^2}{3}
    ,
\label{eq:D2jeqkfinal}
\end{align}
which uses the well-known solution to the Basel problem~\cite{UpdatedAyoub1974}: $\sum_{m=1}^\infty m^{-2} = \pi^2/6$.

Since Eqs.~\eqref{eq:D2jeqkfinal} and~\eqref{eq:D2jneqkfinal} represent the two cases in Eq.~\eqref{eq:D2jk}, we have proven that
\begin{align}
    \mat D_{(2)} &= \mat D^2,
\intertext{and thus, as expected,}
    \partialBL_{xx} &= \partialBL_x \partialBL_x.
\end{align}

\section{Proof of Theorem~\ref{thm:Hquad}}
\label{app:proofofHquad}

\begin{proof}
The Hamiltonian in question, $\op H_{\mathrm{quad}}$, contains only quadratic terms of the form $\op q_j \op q_{j+a}$, $\op q_j \op p_{j+a}$, and $\op p_j \op p_{j+a}$.
Our strategy will be to calculate the commutator of~$\op P$ with each of these terms and then show that the sum over~$j$ for each of these cancels to zero.

We start with the definition of~$\op P$, reproduced from Eq.~\eqref{BLtotmomentum} and using the abbreviation of the coefficients from~Eq.~\eqref{eq:Dmatrix}:
\begin{align}
 \label{eq:BLtotmomentumwithD}
     \op{P} 
&=
     \sum_{i\in\integers}
     \sum_{n\in\integers}
     D_{n0}
     \op{p}_i
     \op{q}_{i+n}
    \nonumber 
\\
&=
     -\sum_{k,l \in \integers}
     D_{kl}
     \op{p}_k
     \op{q}_{l}
     ,
\end{align}
where we have used $D_{kl} = D_{k-l,0} = -D_{lk}$ and reindexed the sum on the second line. 
where $\opvec q$ and $\opvec p$ are column vectors of the associated operators, and $^\tp$ indicates a row vector of operators instead.

To begin, we calculate
\begin{subequations}
\begin{align}
    [\op{q}_j,\op{P}]
&=
    -
     \sum_{k,l \in \integers}
     D_{kl}
     \underbrace{
     [\op{q}_j,\op{p}_k]
     }
     _{i \delta_{jk}}
     \op{q}_{l}
=
    -i
    \sum_{l \in \integers}
    D_{jl}
    \op{q}_{l}
=
    -i
    \op q'_j
    ,
\\
   [\op{p}_j,\op{P}]
&=
    -
    \sum_{k,l \in \integers}
    D_{kl}
    \op{p}_k
    \underbrace{
    [\op{p}_j,\op{q}_{l}]
    }
     _{-i \delta_{jl}}
=
    -i
    \sum_{l \in \integers}
    D_{jl}
    \op{p}_l
=
    -i
    \op p'_j
    ,
\end{align}
\end{subequations}
where $\op q'_j$ and $\op p'_j$ are defined in Eq.~\eqref{qpprimedefs}. 
Compare these with Eqs.~\eqref{eq:comconditions} to confirm the relationship with the continuous version.

We now evaluate the possible quadratic commutators using the above results, along with
 \begin{align}
     [\op{A}\op{B}, \op{C}]=\op{A}[\op{B},\op{C}]+[\op{A},\op{C}]\op{B},
 \end{align}
which holds for any operators $\op{A}$, $\op{B}$, $\op{C}$:
\begin{subequations}
\begin{align}
    [\op{q}_{j} \op{q}_{k},\op{P}]
&=
    -i
    \sum_{l\in\integers}
    \left(
    D_{kl}
    \op{q}_{j}
    \op{q}_{l}
    +
    D_{jl}
    \op{q}_{l}
    \op{q}_{k}
    \right)
    ,
\\
    [\op{q}_{j} \op{p}_{k},\op{P}]
&=
    -i
    \sum_{l\in\integers}
    \left(
    D_{kl}
    \op{q}_{j}
    \op{p}_{l}
    +
    D_{jl}
    \op{q}_{l}
    \op{p}_{k}
    \right)
    ,
\\
    [\op{p}_{j} \op{p}_{k},\op{P}]
&=
    -i
    \sum_{l\in\integers}
    \left(
    D_{kl}
    \op{p}_{j}
    \op{p}_{l}
    +
    D_{jl}
    \op{p}_{l}
    \op{p}_{k}
    \right)
    .
\end{align}
\end{subequations}
All three of these equations are of the form
\begin{align}
    [\op \xi_{j} \op \zeta_{k},\op P]
&=
    -i
    \sum_{l\in\integers}
    \left(
    D_{kl}
    \op \xi_{j}
    \op \zeta_{l}
    +
    D_{jl}
    \op \xi_{l}
    \op \zeta_{k}
    \right)
\nonumber
\\
&=
    -i
    (
    \op \xi_{j}
    \op \zeta'_{k}
    +
    \op \xi'_{j}
    \op \zeta_{k}
    )
    ,
\end{align}
with $\op \xi$ and $\op \zeta$ each standing for either $\op q$ or $\op p$. We recognize this as a bandlimited version of the derivative product rule in the discrete representation.

The quantities we actually want are all of the form
\begin{align}
\label{eq:genformcommutator}
    \sum_{j\in\integers}
    [\op \xi_{j} \op \zeta_{j+a},\op P]
&=
    -i
    \sum_{j\in\integers}
    (
    \op \xi_{j}
    \op \zeta'_{j+a}
    +
    \op \xi'_{j}
    \op \zeta_{j+a}
    )
    .
\end{align}
To evaluate this, we use the following equivalent ways to express the action of the bandlimited derivative:
\begin{align}
\label{eq:xiderivforms}
    \op \xi'_j
&=
    \sum_{l\in\integers}
    D_{jl}
    \op \xi_l
=
    \sum_{m\in\integers}
    D_{m0}
    \op \xi_{j-m}
=
    -
    \sum_{m\in\integers}
    D_{m0}
    \op \xi_{j+m}
    ,
\end{align}
where the second and third sums are reindexed with ${m = j-l}$ and ${m = l-j}$, respectively, and we use the properties of~$D$ discussed below Eq.~\eqref{eq:BLtotmomentumwithD}. Using these, we can expand Eq.~\eqref{eq:genformcommutator} as
\begin{align}
    \label{eq:D10method}
    \sum_{j\in\integers}
    [\op \xi_{j} \op \zeta_{j+a},\op P]
&=
     -i
    \sum_{m,j\in\integers}
    D_{m0}
    (
    \op \xi_{j}
    \op \zeta_{j+a+m}
    -
    \op \xi_{j-m}
    \op \zeta_{j+a}
    )
\nonumber
\\
&=
     -i
    \sum_{m,j\in\integers}
    D_{m0}
    (
    \op \xi_{j}
    \op \zeta_{j+a+m}
    -
    \op \xi_{j}
    \op \zeta_{j+a+m}
    )
\nonumber
\\
&=
    0
    ,
\end{align}
where we have used one form each from Eq.~\eqref{eq:xiderivforms} in the first line, and we have reindexed the $j$ sum ($j \mapsto j+m$) in the second. Thus,
\begin{align}
    \sum_{j\in\integers}
    [\op q_{j} \op q_{j+a},\op P]
&=
    0
    ,
\\
    \sum_{j\in\integers}
    [\op q_{j} \op p_{j+a},\op P]
&=
    0
    ,
\\
    \sum_{j\in\integers}
    [\op p_{j} \op p_{j+a},\op P]
&=
    0
    ,
\end{align}
as we wanted to prove.

This process shows that, due to the antisymmetry of the coefficients of the sum over $m$ and the freedom to re-index the sum over $j$, every term in the summation will have a counterpart of opposite sign. As a result, the total sum will be zero. 
Finally, as all of these are equal to zero, any linear combination of them will also be equal to zero. As such, any lattice-translationally invariant quadratic Hamiltonian~$\op H_{\mathrm{quad}}$ will commute with~$\op{P}$. And thus, the bandlimited total momentum~$\op{P}$ is conserved in such a system.
\end{proof}
 
While $\op{P}$ commutes with (lattice-translationally invariant) Hamiltonians of quadratic order, the same is not always true for Hamiltonians of polynomial order greater than two.
One can check this by calculating
\begin{align}
    \sum_{j\in \integers}
    [\op{q}_j^3, \op{P}]
&= 
    \sum_{j\in\integers}
    \left(
    \op{q}_j^2
    [\op{q}_j,\op{P}]
    +\op{q}_j
    [\op{q}_j, \op{P}]
    \op{q}_j
    + [\op{q}_j, \op{P}]
    \op{q}_j^2
    \right)
\nonumber \\
&=
    -i\sum_{j\in \integers}
    3\op{q}_j^2
    \op{q}_j'.
    \label{eq:cubiccommutator}
\end{align}

Like the product rule acknowledged above, this is an example of a bandlimited version of the derivative chain rule in the discrete representation.
Using the properties of~$D$ discussed below Eq.~\eqref{eq:BLtotmomentumwithD}, we can rewrite Eq.~\eqref{eq:cubiccommutator} using a process similar to the one shown in Eq.~\eqref{eq:D10method}:
\begin{align}
    \label{eq:D16method}
      \sum_{j\in \integers}
    [\op{q}_j^3, \op{P}]
    &=  
    -3i
    \sum_{m,j\in\integers}
    D_{m0}
    \op{q}_j^2
    \op{q}_{j+m}
\\*
    &=
    -3i
    \sum_{m,j\in\integers} 
    \frac{1}{2}
    D_{m0}
    \left(
    \op{q}_j^2
    \op{q}_{j+m}
    -
    \op{q}_{j+m}^2
    \op{q}_j
    \right).
\nonumber 
\end{align}
However, unlike the case of Eq.~\eqref{eq:D10method}, the sums over $m$ and $j$ cannot be reindexed such that the terms inside the brackets  of Eq.~\eqref{eq:D16method} cancel to zero. 
As a result, we have
\begin{align}
    \sum_{j\in \integers}[\op{q}_j^3, \op{P}]\neq 0.
\end{align}
Crucially, this means that while Hamiltonians of quadratic power will always commute with the total momentum operator of the field, the same is not true for Hamiltonians of polynomial power greater than two, indicating that interacting fields on a lattice may not possess continuous translational invariance, unlike their free field counterparts.

\newpage
\bibliographystyle{IEEEtran}
\bibliography{refs}

\end{document}